\documentclass[12pt]{article}
\pdfoutput=1
\usepackage[style=apa,backend=biber,natbib=true,uniquename=false]{biblatex}
\newcommand*\samethanks[1][\value{footnote}]{\footnotemark[#1]}

\usepackage{amsmath}
\usepackage{amssymb}
\usepackage{graphicx,mwe}
\graphicspath{ {./images/} }
\usepackage{authblk}
\usepackage{amsthm}
\usepackage{float}
\usepackage{faktor}
\usepackage{mathtools}
\usepackage{diagbox}
\usepackage{soul}
\usepackage[dvipsnames]{xcolor}
\usepackage{bm}
\usepackage{algorithm}
\usepackage{breqn}
\usepackage{algpseudocode}
\usepackage{setspace}
\usepackage{makecell}
\usepackage{multirow}
\usepackage{comment}
\usepackage[title]{appendix}
\usepackage[a4paper, portrait, margin=1in]{geometry}
\usepackage{framed}
\usepackage[flushleft]{threeparttable}
\usepackage[most]{tcolorbox}
\usepackage{xcolor}
\colorlet{shadecolor}{orange!15}
\newtheorem{definition}{Definition}

\newtheorem{theorem}{Theorem}

\usepackage{enumerate}
\usepackage{hyperref}
\usepackage{tikz}
\usetikzlibrary{arrows}
\makeatletter
\newsavebox\myboxA
\newsavebox\myboxB
\newlength\mylenA
\usepackage{graphicx} 
\usepackage{caption,subcaption}

\providecommand{\keywords}[1]
{
  \small	
  \textbf{\textit{Keywords---}} #1
}

\usepackage{xcolor}
\definecolor{light-gray}{gray}{0.95}
\newcommand{\code}[1]{\colorbox{light-gray}{\texttt{#1}}}

\newcommand{\E}[1]{\mathbb{E}\left[#1\right]}
\newcommand{\Prob}[1]{\mathbb{P}\left(#1\right)}

\newcommand*\xoverline[2][0.75]{%
    \sbox{\myboxA}{$\m@th#2$}%
    \setbox\myboxB\null
    \ht\myboxB=\ht\myboxA%
    \dp\myboxB=\dp\myboxA%
    \wd\myboxB=#1\wd\myboxA
    \sbox\myboxB{$\m@th\overline{\copy\myboxB}$}
    \setlength\mylenA{\the\wd\myboxA}
    \addtolength\mylenA{-\the\wd\myboxB}%
    \ifdim\wd\myboxB<\wd\myboxA%
       \rlap{\hskip 0.5\mylenA\usebox\myboxB}{\usebox\myboxA}%
    \else
        \hskip -0.5\mylenA\rlap{\usebox\myboxA}{\hskip 0.5\mylenA\usebox\myboxB}%
    \fi}
\makeatother
\title{Estimating Policy Effects in a Social Network with Independent Set Sampling}

\author{Eugene T.Y. Ang\thanks{Institute of Operations Research and Analytics (IORA), National University of Singapore (NUS)} , Prasanta Bhattacharya\thanks{Institute of High Performance Computing (IHPC), Agency for Science, Technology and Research (A*STAR), 1 Fusionopolis Way, \#16-16 Connexis, Singapore 138632, Republic of Singapore} , Andrew E.B.  Lim\samethanks[1] \thanks{Department of Analytics and Operations, National University of Singapore (NUS)} \thanks{Department of Finance, National University of Singapore (NUS)} \thanks{Andrew E.B. Lim is supported by the  Ministry of Education, Singapore, under its 2021 Academic Research Fund Tier 2 grant call (Award ref: MOE-T2EP20121-0014).}}

\date{}

\begin{document}

\maketitle

\begin{abstract}
Evaluating the impact of policy interventions on respondents who are embedded in a social network is often challenging due to the presence of network interference within the treatment groups, as well as between treatment and non-treatment groups throughout the network. In this paper, we propose a novel empirical strategy that combines network sampling based on the identification of independent sets with a stochastic actor-oriented model (SAOM) to infer the direct and net effects of a policy. By assigning respondents from an independent set to the treatment, we are able to block direct spillover of the treatment among the treated respondents for an extended period of time, during which the direct effect of the treatment can be isolated from the associated network interference. We empirically demonstrate this using a simulation-based evaluation of a fictitious policy implementation using both real-life and generated networks, and use a counterfactual approach to estimate the treatment effect of the policy. Our results highlight the effectiveness of our proposed empirical strategy, and notably, the role of network sampling techniques in influencing the evaluation of policy effects. The findings from this study have the potential to help researchers and policymakers with planning, designing, and anticipating policy responses in a networked society. 
\end{abstract}

\keywords{Independent set sampling, Stochastic actor-oriented model, Policy evaluation}

\newpage

\section{Introduction}
Evaluating the impact of policies is critical to good governance in both offline and online spaces \parencites{athey2017, mergoni2021, sanderson2002evaluation}. Policymakers need to assess if a policy has achieved its intended outcomes and identify the factors that contributed to, or hindered its effectiveness. Accurate estimation of policy outcomes is therefore critical to the design and implementation of future ones. To estimate policy impact, evaluation studies use either experimental or quasi-experimental approaches \parencites{gertler2016impact, shadish2002experimental}{miller2020experimental} that involve a controlled (or quasi-controlled) exposure of groups to the specific intervention being studied. However, \emph{network interference} can be a major obstacle in such studies because the targets of policy interventions are often socially embedded entities, such as individuals or organizations \parencites{huggins2001inter,borgatti2003network}{GRAF2020}, that do not exist in isolation \parencites{basse2017}{viviano2022policy}. When policies are implemented in a networked population there is a possibility that entities within the network are influenced by the behavior of their peers. 
For example, economic policies to boost investment can accelerate job creation in targeted sectors, and improve consumption levels. However, such policies can also have unintended spillover effects, both positive and negative, in other sectors of the economy  (\cite{liu2008foreign, crespo2007determinant}). This can lead to amplification or mitigation of the intended policy outcomes, and result in inaccurate policy or business decisions in future. To avoid such uncertainty in estimating outcomes, policymakers need to account for any network effects stemming from the policy \parencite{forastiere2024causal}. In this study, we introduce an empirical strategy to aid researchers and policymakers evaluate the effects of a policy on the target behavior by decoupling the \textit{direct effect} of the policy from \textit{indirect effects} stemming from network interference within the treated group. \\

As policymakers are often faced with resource and feasibility constraints, they are unable to implement and evaluate interventions on the entire population of interest. Hence, a key consideration in policy evaluation studies is the selection of an appropriate sampling technique for assigning individuals to the intervention. This consideration is particularly salient since the choice of the initial seeding or sampling strategy determines the resulting network structures. This, in turn, leads to different patterns of network interference and spillover effects, which then affect the assessment of the policy treatment effect \parencites{viviano2024policy, aronow2017, schwarz2021}. 
In this study, we propose a novel network sampling and assignment strategy as part of a broader policy evaluation methodology. This technique is agnostic to both context and policy type and can hence be applied flexibly to a wide range of societal and organizational policies. It is based on the identification of an independent set that includes respondents who are not directly connected to one another via a contact network. By selectively exposing this independent set to the treatment, we are able to block any treatment spillover via peer influence among policy-exposed respondents for an extended period of time after the policy intervention. This can provide policymakers with an opportunity to decouple the direct policy effects from the indirect network-induced effects. \\ 

In this study, we determine the relative efficacy of these sampling strategies by comparing the \emph{time to first change} (TTFC), which we define as the time taken for the focal behavior of the policy-exposed respondents to change \textit{after} the policy implementation. 
As we show later in this paper, TTFC offers a way to assess the time taken for network interference to alter the behavior of the treatment group. So, a sampling strategy with a larger TTFC within the treatment group will be better suited for estimating the direct effect of the policy. We run an illustrative policy implementation exercise based on a real-world dataset, and use a stochastic actor-oriented model (SAOM) to compare the empirical average TTFC across different sampling techniques and network structures. Through this exercise, we demonstrate that the independent set sample experiences a delayed onset of network interference when compared to other types of sampling strategies. Subsequently, we show that the direct effect of the policy can be better decoupled from the associated network effect at specific time thresholds based on these empirically estimated TTFC durations. \\

We also highlight the robustness of our proposed independent set-based policy evaluation method to variations in independent set composition and network structure. For instance, past studies have highlighted the structural importance of considering centrality measures in sampling strategies (e.g., \cite{maiya2010online}). We incorporate several centrality measures, such as degree and closeness centralities in the construction of our independent set sample, and provide a comparative analysis of their effectiveness. We also replicate our method on contact network structures produced using a range of random graph-based models (e.g., Erd\"{o}s-Renyi, Watts-Strogatz and scale-free network models) to highlight the robustness of our findings to changes in topology. In our simulation exercises, we observe that the network size also contributes to the effectiveness of our proposed sampling strategy. For instance, when the underlying contact network is very sparse, sampled respondents have very few connections among them irrespective of the sampling strategy, leading to little differentiation in model performance among them. \\

In the following section, we discuss past work on estimating policy impact within networked contexts, and offer a comparative analysis of prior empirical methods for policy evaluation. Subsequently, in Section \ref{sec:methodology}, we introduce our proposed method, and the mechanics of our simulation-based models to infer treatment effects. In Section \ref{sec:analysis}, we implement this method on a real-world social network dataset to demonstrate how the direct effect of a (fictitious) policy can be estimated in the presence of network interference. Additionally, we present a set of robustness analyses to test the sensitivity of our method to changes in sample composition and network structure. Finally, in Section \ref{sec:discussion}, we conclude by summarizing the key findings from this study, some limitations of our current approach, and the plan for future extensions of this work.

\section{Policy Evaluation in Networked Contexts}\label{sec:relatedwork}
Randomized controlled trials (RCT) are often considered to be the gold standard among methods to evaluate interventions \parencite{Hariton2018}. However, RCTs are often difficult to implement within a social network due to methodological challenges such as network interference \parencites{karwa2018}{schwarz2021}, and other practical or ethical concerns \parencite{goldstein2018ethical}. Hence, researchers often rely on quasi-experimental or observational studies to identify the causal effects of an intervention. \\

Recent studies on policy interventions in networked contexts have leveraged the dynamics of interpersonal connections and behavior to address key outcomes such as managing the effects of COVID-19 \parencite{robins2023multilevel}, managing natural resources to mitigate climate change \parencite{frank2023network}, and fostering scientific collaboration and productivity \parencite{sciabolazza2020connecting}. Tackling the issue of network interference in policy evaluation can often require large and complex experiment designs \parencite{athey2017}. A number of past studies have formulated randomization-based frameworks for estimating causal effects 
\parencites{aronow2017}{baird2018}. 
Others have attempted to estimate such effects 
using non-parametric and regression-based estimators under specific distributional assumptions (e.g., \cite{leung2020treatment}). 
Several network randomization studies make use of cluster sampling \parencite{hu2013}, which offers a way to assign entire communities of individuals to treatment groups, as opposed to conventional techniques that randomize at an individual level. This method of randomizing the assignment of clusters to experimental conditions can help reduce bias in the focal estimates \parencites{ugander2013}{EcklesKarrerUgander+2017}{ugander2020}. \\

However, the assignment of respondents to different exposure groups is non-trivial in a networked setting. Even if we are able to generate clusters that minimize interactions across groups (e.g., \cite{blondel2008,FORTUNATO2010}), 
the stable unit treatment value assumption (SUTVA) might be violated due to the presence of network interference \textit{within} the groups \parencite{ugander2013}. In other words, we cannot determine if changes in individuals' behavior are solely due to their own treatment or due to influence from their treated peers. 
This inherent difficulty in eliminating network interference within the treated clusters might lead to a biased estimate of the policy effect, as well as other network effects \parencite{basse2017}. Furthermore, it is difficult to quantify the degree of interference in a networked context. Studies often make prior assumptions on the extent of interference and the properties of the underlying network, but this can limit the generality of their approaches. 
Our work contributes to the current policy evaluation literature in a number of ways. Specifically, our proposed approach is context- and policy-agnostic and does not require strong assumptions on the type of treatment or network structure to estimate and decouple the direct policy effect from the network interference, within the treated group. This makes it appropriate for use in a wide range of organizational and societal policy contexts and provides policymakers greater flexibility to estimate policy effects. Our proposed strategy does, however, require the policy implementer to have access to the entire contact network as well as an understanding or intuition of the potential network effects, which is instrumental for specifying the model, as we explain later.  \\


We also note that accounting for network effects in \textit{observational studies} is complicated by various empirical challenges, such as the difficulty in separating social influence from other confounders like homophily and shared contexts \parencite{shalizi2011}. Previous studies have attempted to address these limitations using various methods such as linear-in-means models \parencites{manski1993,kline2012, blume2015linear}, stochastic actor-oriented models (SAOM) \parencites{snijders1996stochastic,snijders2007,snijder_2010} and exponential random graph models (ERGM) \parencites{wasserman1996}\footnote{Including extensions such as longitudinal and separable temporal exponential random graph models \parencites{krivitsky2019stergm, koskinen2015simultaneous}}.  We note that an important limitation of the linear-in-means model is its use of aggregated endogenous variables \parencite{manski1993}. This reduces available information on alters' characteristics, as well as the strength and direction of peer effects that often vary as a function of such individual-level characteristics and interaction contexts. Moreover, earlier studies using the linear-in-means model have often assumed that the social network structure is exogenous and static, which is a strong assumption for many real-world social networks that tend to be inherently dynamic. Hence, such linear-in-means models are not well suited to account for the joint evolution of individual behavior and network change. We note, however, that recent studies have used variants of the model that consider the endogeneity of the network structure \parencites{johnsson2017}{JOCHMANS2022}. Nonetheless, such models rely on strong assumptions to tackle endogeneity, such as the validity of certain instrumental variables. Moreover, the linear-in-means model assumes that endogenous variables are additive and linear. However, peer effects can often depend on the characteristics and behaviors of multiple alters in complex and non-linear ways. In such contexts, a linear model might generally not be appropriate for estimating the underlying peer effects.\\

For modeling complex mechanisms on dynamic social networks, researchers have commonly used the SAOM as well as temporal variants of the exponential random graph model (e.g., STERGM, TERGM). Other event-based models (e.g., \cite{butts2008relational, perry2013point, stadtfeld2011analyzing}) and dynamic network actor models (DyNAM) \parencites{stadtfeld2011analyzing, stadtfeld2017dynamic} can also be used for modeling instantaneous interactions in temporal networks. However, since the current study focuses on estimating policy effects after a set period of time, we contend that ERGM- or SAOM-based models are better suited for our case. Further, in selecting between an ERGM or a SAOM for our proposed method, Block et. al. (2016) recommend employing one whose assumptions better fit the given social and network processes \parencite{block2016}. In our context, we investigate the dynamic or co-evolving changes in individual behavior and network ties due to policy implementation, and hence, the SAOM is preferred over the standard ERGM. This is consistent with past studies which have incorporated SAOMs to generate insights that can better inform policy interventions by analyzing the mechanisms (e.g., peer influence) which drive social outcomes (e.g., adolescents' drinking behavior) (e.g., \cite{ivaniushina2019peer}). Hence, in our work, we propose an empirical strategy that integrates sampling techniques with an SAOM to estimate the direct and net effect of a policy. Specifically, by simulating the network dynamics using a SAOM, we are able to track the change in focal behavior adoption among the treated respondents in the post-policy period, which allows us to detect the onset of any potential network interference.\\

Lastly, an important consideration in policy evaluation studies is the propagation of policy effects in a social network, where the spread of policy-linked behavior can be modeled as a contagion process. Several studies have investigated methods to maximize the spread of influence and ideas through carefully selected nodes (e.g., \cite{kempe_2003, valente1999accelerating}). Although the current paper focuses on the estimation of policy effects and not on forecasting contagion-based policy outcomes, our method of using the time to first change (TTFC) to disentangle the direct policy effects from network interference shares some similarities with how contagion or social influence processes have been modeled in the literature. 

\section{Proposed Method}\label{sec:methodology}
In a typical policy evaluation exercise, researchers are primarily interested in measuring the \emph{direct} policy effect, i.e., the behavioral change of a respondent as a result of being exposed to the policy. However, as discussed earlier, the presence of network interference among the policy-exposed respondents can confound the measurement of this direct effect. Network interference can manifest via two indistinguishable forms of \emph{indirect effects}, namely, the \emph{indirect network effect} and \emph{indirect policy-linked network effect}, which are the behavioral changes of a respondent caused by peer influence from an unexposed vs. exposed respondent, respectively. In this section, we present our empirical strategy to disentangle the direct effect of the policy intervention from the indirect effects via the social network. We begin by introducing the independent set sampling technique, and explain how it can be suitably incorporated into our evaluation method by selectively exposing respondents in the independent set to the policy. Subsequently, we explain how we can combine this sampling and assignment procedure with a stochastic actor-oriented model (SAOM) to estimate the direct effect of a policy, in the presence of network-behavior co-evolution. 

\subsection{Independent set sampling}\label{sec:indpsetsampling}
Given access to a social or contact network of the policy-relevant population, we can compute an independent set of respondents from the population, as defined below.
\begin{definition}[Independent Set]
A set of nodes $S$ is called an \textbf{independent set} if no two vertices in this set $S$ are adjacent to each other in the given network.
\end{definition}
For an independent set of respondents, these selected respondents are not connected to one another. We provide an example of selecting independent sets from a cyclic graph of 6 nodes, denoted as $C_6$, in Figure \ref{Fig1}.

\begin{figure}[H]
\centering
\begin{subfigure}[b]{0.4\linewidth}
\begin{center}
\begin{tikzpicture}[y=.3cm, x=.3cm,font=\small] \draw [latex-latex](0,4) -- (3,8); \draw [latex-latex](8,8) -- (11,4); \draw [latex-latex](3,0) -- (8,0); \draw [latex-latex,](0,4) -- (3,0); \draw [latex-latex,](8,0) -- (11,4); \draw [latex-latex](3,8) -- (8,8); \filldraw[fill=black!40,draw=black!80] (0,4) circle (5pt) ; \filldraw[fill=white!40,draw=black!80] (3,0) circle (5pt) ; \filldraw[fill=white!40,draw=black!80] (3,8) circle (5pt) ; \filldraw[fill=white!40,draw=black!80] (8,0) circle (5pt) ; \filldraw[fill=white!40,draw=black!80] (8,8) circle (5pt) ; \filldraw[fill=black!40,draw=black!80] (11,4) circle (5pt) ; 
\end{tikzpicture}
\end{center}
\caption{Maximal independent set}\label{maximalindependentset}
\end{subfigure}
\begin{subfigure}[b]{0.4\linewidth}
\begin{center}
\begin{tikzpicture}[y=.3cm, x=.3cm,font=\small] \draw [latex-latex](0,4) -- (3,8); \draw [latex-latex](8,8) -- (11,4); \draw [latex-latex](3,0) -- (8,0); \draw [latex-latex,](0,4) -- (3,0); \draw [latex-latex,](8,0) -- (11,4); \draw [latex-latex](3,8) -- (8,8); \filldraw[fill=black!40,draw=black!80] (0,4) circle (5pt) ; \filldraw[fill=white!40,draw=black!80] (3,0) circle (5pt) ; \filldraw[fill=white!40,draw=black!80] (3,8) circle (5pt) ; \filldraw[fill=black!40,draw=black!80] (8,0) circle (5pt) ; \filldraw[fill=black!40,draw=black!80] (8,8) circle (5pt) ; \filldraw[fill=white!40,draw=black!80] (11,4) circle (5pt) ; 
\end{tikzpicture}
\end{center}
\caption{Maximum independent set}\label{maximumindependentset}
\end{subfigure}
\caption{Selecting the shaded nodes from a $C_6$ graph to form different independent sets}\label{Fig1}
\end{figure}

Depending on the selection procedure, we can obtain different independent sets of varying sizes. In Figure \ref{maximalindependentset}, the independent set obtained is maximal as a further selection of any nodes breaks the independence criterion. However, this set is not the largest, as shown in Figure \ref{maximumindependentset}. We note that there is an upper bound on the maximality of the size of the independent set, attributed to Kwok \parencite{west1996introduction}, as stated in the following theorem.

\begin{theorem}[Kwok Bound]
Let $G$ be a graph on $n$ vertices with $e$ edges and let $\Delta$ be the maximum vertex degree. Then, the independence number $\alpha(G)$, which is the size of the largest independent set of $G$, has the following upper bound,
$$\alpha(G) \leq n - \frac{e}{\Delta}$$
\end{theorem}

There is also a lower bound on the size of the independent set, as attributed to Caro and Wei \parencites{caro1979new}{wei1981lower},

\begin{theorem}[Caro-Wei Bound]
Let $G = (V, E)$ be a graph and let $d_v$ be the degree of vertex $v$. Then, the independence number $\alpha(G)$, which is the size of the largest independent set of $G$, has the following lower bound,
$$\alpha(G) \geq \sum_{v \in V}\frac{1}{1+d_v}$$
\end{theorem}

With this lower bound, policymakers can have a guarantee on the size of the largest independent set sample. These bounds provide them the flexibility of choosing a suitable sample size, given their resource constraints and study objectives. We note that finding the maximum independent set is NP-hard \parencite{garey1979computers} and is computationally expensive, especially for large social networks. Instead, one can find a maximal independent set much more efficiently. Apart from the vanilla independent set sample, where selected respondents are not connected to each other, policymakers can also choose to obtain ``stronger" independent sets, where the distance between pairwise sampled nodes is at least $k$, for $k \geq 2$. However, we note that there exists an inherent trade-off between the efficacy of the sampling technique and sample size. To obtain a reliable estimate of the policy effect, policymakers might require a large sample size which is resource-intensive to obtain. On the other hand, a modest sample size which can optimize resource allocation may compromise statistical power or the detectable effect size. Although ``stronger" independent sets can be useful in complex contagion studies (e.g., \cite{centola2007complex}), our current study uses the independent set sampling strategy as part of a policy evaluation method, and hence we focus our analysis on the vanilla independent set.\\

As social networks are finitely big and the maximum degree of any social network is smaller than the number of nodes, we are able to find an independent set efficiently \parencite{Halldórsson1997}. We note that there is a possibility of isolated nodes present in the network, but this does not affect the sampling procedure. Since isolated nodes are not connected to any other nodes, we can choose to either include or exclude them in the treated sample, depending on the design choices of the policy evaluation study. 

\subsection{Specifying a stochastic actor-oriented model}\label{sec:specifySAOM}

Following the construction of the independent set-based sample, we specify and fit a stochastic actor-oriented model (SAOM) \parencites{snijders1996stochastic, snijders2007, snijder_2010} to track how the network structure evolves in response to changes in behavior, and vice versa, after accounting for individual-level (e.g., age, gender) covariates. We specify the technical details of the SAOM model, specifically the rate and objective functions, in Appendix \ref{appendix:SAOMdetails}. As the SAOM specifies the co-evolution using a continuous-time Markov process, we are able to inspect the network-behavior states at every simulated \textit{micro-step} between the observed states. Specifically, this allows us to track the changes in focal behavior of the sampled respondents who are likely to be \emph{directly} influenced by the behavior of their peer respondents \parencite{steglich2010}, as well as the broader social network. By inspecting the specific micro-steps when the behavior change occurs, we can assess the onset of network interference in the treatment set. Furthermore, the SAOM is flexible and versatile as it allows policymakers to specify a wide variety of network and behavioral processes based on their domain expertise. The estimated model can then be used to simulate new networks under different counterfactual scenarios related to changes in the fitted network and behavioral processes. In the next section, we illustrate how the independent set sampling strategy as a selection procedure can be used together with the SAOM, as described here, to isolate the direct effect of a policy from the indirect (policy) network effects.

\subsection{Estimating policy effects using SAOM} \label{sec:estimateeffect}
As mentioned earlier, the sampled respondents exposed to the policy may influence each other via their underlying contact network, which can then potentially alter their focal behavior. Due to the design of the independent set, policy-exposed respondents are not connected to one another, and hence unable to directly influence each other's behavior. This provides the policymakers with an extended temporal window to estimate the direct effect of the policy. However, we are unable to obtain a complete and long-term isolation of the policy-exposed respondents in a social network due to the gradual onset of the network interference over time. Thus, the key consideration here is that the direct effect of a policy on the sampled respondents, who are selected through the independent set sampling technique, can be better isolated in the period immediately following the policy implementation. To empirically illustrate this point, we introduce the concept of \emph{time to first change} (TTFC) as the duration in the post-policy period when one of the sampled respondents first experiences a behavior change. Formally, suppose that there are $n$ sampled respondents who are exposed to the policy at time $t=0$. Each policy-exposed respondent has a binary behavior variable $X_i(t)$, for $i \in [n]$, and $\sum_i X_i(t)$ represents the total adoption of the focal behavior at time $t$. Assuming that there is no time lag in response to the policy i.e., the direct effect of the policy manifests instantaneously, the time to first change (TTFC) is a random variable, that is given as 
$$\text{TTFC} = \inf\{t \, | \, \sum_i X_i(t) \neq \sum_iX_i(0)\}$$
The policymakers can then obtain an empirical estimate of the TTFC by tracking the micro-steps through the SAOM, as described in the earlier subsection. Since the TTFC is sensitive to the network centralities and the ``position" of the sampled respondents in the network, we repeat the simulations from the SAOM multiple times to obtain an empirical average TTFC. It is important to highlight here that the TTFC captures \textit{any} change that is not directly attributable to the policy, i.e., it includes changes due to network interference as well as the individuals' natural tendencies towards changing their behavior over time. 
As the direct treatment effect is confounded with the respondents' natural tendency to change behavior, we estimate the direct effect by adopting a counterfactual approach. Specifically, we re-estimate the SAOM using the same networks and behavior, but \textit{without} the policy intervention or the policy-induced behavior changes. By tracking the same group of sampled respondents across both scenarios and taking the absolute difference in the proportion of sampled individuals who possess the focal behavior, we cancel out any effect of the respondents' natural tendency. Using suitable temporal cut-offs based on the TTFC, we are able to isolate and estimate the direct effect in the pre-TTFC period. Policymakers can choose appropriate cut-offs based on their domain knowledge and experience. In the following section, we illustrate this process through an empirical study using both real-world and generative network data. 

\section{Policy Intervention Analysis under Network Interference}\label{sec:analysis}
In this section, we empirically illustrate the effectiveness of our proposed method using a simulation-based study. We leverage a pre-existing and well-known social network dataset to mimic a real-world policy introduction and evaluation exercise. Since there are numerous ways to select an independent set sample, we perform a comparative analysis of several notable options. Furthermore, we illustrate the robustness of our method under various network structures, created using various graph generation models. 

\subsection{Illustrative study using a real-world social network}\label{sec:simulationstudy}
\textbf{Data.} For this simulation exercise, we use data from the \emph{Teenage Friends and Lifestyle Study} \parencites{michell2000smoke, pearson2003drifting}. The dataset tracks a cohort of students from 1995 to 1997, whose data on friendship networks, smoking behavior, and other lifestyle variables were measured at three successive time points. Although the social networks in the dataset were obtained by asking respondents to nominate up to six friends and rank the strength of these friendships, we use an unweighted and undirected version of this network for our study. We note that the dataset faces left and/or right censoring, with 129 out of the 160 respondents present throughout the cohort study. We include these 129 respondents for our simulation exercise. The original study provides data on the frequency of the respondents' substance use and leisure activities. For our study, we treat smoking as our focal and binary behavior where respondents are considered smokers if they smoke within the observation period. We also include other socio-demographic attributes such as age and gender, as well as income which is measured in the study by the amount of pocket money received. \\

\textbf{Sampling process and robustness checks.} We select respondents for policy exposure using our proposed independent set sampling technique, as well as using random and cluster sampling techniques in order to test the relative effectiveness of our proposed strategy. We first iteratively generate a maximal independent set (see Algorithm \ref{alg: indpset} in Appendix \ref{appendix:selection} for details). To ensure that the size of the set of policy-exposed respondents does not affect the estimation of the expected TTFC as described in the previous section, we try to maintain a consistent sample size across the three sampling techniques. We obtain a set of randomly selected respondents that is of equal size to the independent set sample. To obtain our cluster sample, we use the popular fast greedy modularity algorithm \parencite{clauset2004finding}, which has been used in recent studies \parencites{nadini2021mapping, hernandez2021environmental}. We note that it is difficult to ensure that the cluster sample is the same size as the independent set and random samples. Hence, we choose clusters whose total size is close to the size of the independent set sample, with a difference of at most 20 respondents. We note that it is possible for policymakers to selectively expose only a subset of these samples to the policy, due to resource constraints or other factors. However, for the current exercise, we expose the entire sample to the policy. We provide a graph visualization of the possible samples obtained through the various strategies and corresponding algorithms in Figure \ref{fig: graphvisualization} in Appendix \ref{appendix: samplevisualization}.\\

Another consideration in our proposed approach is the question of structural preservation across multiple runs of the independent set sampling technique. To test the robustness of picking any arbitrary maximal independent set for our simulation study, we provide a descriptive summary of common measures such as mean degree, mean transitivity, and pairwise Jensen-Shannon Divergence for network structures obtained from 500 runs of the sampling technique. The results, presented in Appendix \ref{appendix:maxindpset}, show relatively small standard deviations of mean degree and mean transitivity across the runs. Additionally, the low divergence score suggests that the degree distributions across the runs are also very similar. Taken together, this illustrates that the iterative selection of independent sets is robust in preserving important structural properties of the network.\\

\textbf{SAOM model and policy implementation.} Next, we describe our proposed strategy for estimating the average TTFC, as described in the previous section. The \emph{Teenage Friends and Lifestyle Study} dataset can be used to infer the dynamics of the social network and focal behavior \textit{prior} to the policy implementation. To do this, we run a stochastic actor-oriented model (SAOM) to estimate the network and behavioral rates and statistics based on the first three waves (see \cite{snijder_2010, Snijders2017, kalish2020stochastic} and Appendix \ref{appendix:SAOMdetails} for further details on the model specification). We include a set of commonly used statistics, as shown in Table \ref{table:networkestimatesSAOM}, to model the existing dynamics in the dataset, but also emphasize that this choice is purely illustrative. In an actual policy evaluation, researchers and policymakers can collectively determine the appropriate set of statistics to specify for the contact network based on their domain expertise. Apart from our focal statistics that measure peer influence and homophily, we also include income homophily as a control, degree-related popularity effect to account for any preferential attachment in network formation, and the effect of income on smoking. The model estimates are presented in Table \ref{table:networkestimatesSAOM}, and the overall maximum convergence
ratio is 0.139.
\begin{table}[H]
    \centering
    \begin{tabular}{|l|l|}
    \hline
    \textbf{Parameters} & \textbf{Coefficients} \\
    \hline
    Friendship rate (Period 1) & 5.1253  (0.4829) \\
    \hline
    Friendship rate (Period 2) & 4.1124 (0.3822)\\
    \hline
    Degree (Density) & -0.2466 (0.2360)\\
    \hline
    Transitive Triads & 1.1438  (0.0583)\\
    \hline
    Degree of alter & -0.3907 (0.0473)\\
    \hline
    Ego behavior similarity & -0.1677 (0.2079)\\
    \hline
    Behavior similarity & 0.6617 (0.1962)\\
    \hline
    Ego income similarity & 0.0255 (0.0100)\\
    \hline
    Income similarity & 1.2857 (0.4023)\\
    \hline
    Behavior rate (Period 1) & 0.8901 (0.2873)\\
    \hline
    Behavior rate (Period 2) & 0.7167  (0.2278)\\
    \hline
    Behaviour Tendency (Linear Shape) & 0.2326 (1.5475)\\
    \hline
    Average Peer Influence & 5.0589 (1.8821)\\
    \hline
    Behavior Outdegree & 0.2926 (0.3007)\\
    \hline
    Behavior: Effect from Income & 0.0584 (0.0490)\\
    \hline
    \end{tabular}
    \caption{Network and behavioral estimates}
    \label{table:networkestimatesSAOM}
\end{table}
We incorporate these statistics in the fitted model to simulate network and behavioral dynamics post-policy implementation, as we describe next. Specifically, policymakers can select different samples based on the above-mentioned approaches and expose them to the intended policy e.g., an anti-smoking intervention. This aims to alter the focal behavior of the respondents via the direct policy effect and the net effect which includes the peer influence from the social network. For our simulation, we specify that all policy-exposed respondents lose their focal behavior (i.e., smoking status) with a probability of 0.5. We also introduce a small perturbation in the network, specifically a 0.1\% probability of forming or deleting ties, to proxy for any unobserved change in the network that may have occurred during the data collection, sampling and/or policy implementation stages. We note that the choice of perturbation probability is purely illustrative and policymakers can adjust this noise factor based on their domain knowledge or any available contextual data. After the policy implementation, we simulate future networks and behavioral dynamics using the SAOM model specified above\footnote{Our codes are adapted from \url{https://www.stats.ox.ac.uk/~snijders/siena/siena_scripts.htm}, specifically WorkOnChains.R and ShowBehavioralChains.R}. As the income variable in the pre-policy implementation is exogenous and dynamic, we assume for simplicity that income does not change post-policy implementation. \\

\textbf{Deriving TTFC.} To quantify the effectiveness of our proposed strategies, we track the total adoption of the focal behavior within the group of policy-exposed respondents over time and obtain an empirical average TTFC based on 1000 runs of the simulation. We note that the average sample size across the three sampling strategies is 41 over the 1000 runs. As mentioned earlier, the TTFC can be affected by other factors such as the individuals' natural tendency to change their focal behavior over time. However, since we simulate future co-evolutions using the same statistics and corresponding estimates stated in Table \ref{table:networkestimatesSAOM} across the three samples (i.e., independent vs. random vs. cluster), these individual tendencies should be identical across these different samples of similar size. Furthermore, and as we show later in this section, our use of a counterfactual network strategy to estimate the direct policy effect cancels any effect from the individuals' natural tendencies. Thus, our empirical design guarantees that we are able to attribute any differences in TTFC to the differences in network interference generated by the respective sampling techniques. We illustrate the differences in average TTFC among the sampling conditions, their respective standard errors, and the pairwise Kolmogorov-Smirnov (K-S) test statistics in Tables \ref{tab:realdatastoppingtime} and \ref{tab:realdatastoppingtimekstest} below.

\begin{table}[H]
    \centering
    \begin{tabular}{|c|c|c|}
    \hline
    Sampling technique & Empirical average \emph{time to first change} (epoch) & Standard error \\
    \hline
    Independent & 439.477 & 6.781 \\ 
    \hline
    Random & 420.752 & 7.000 \\ 
    \hline
    Cluster & 416.451 & 7.016 \\ 
    \hline
    \end{tabular}
    \caption{Time to first change (TTFC) across three sampling strategies}
    \label{tab:realdatastoppingtime}
\end{table}

\begin{table}[H]
    \centering
    \begin{tabular}{|c|c|c|c|}
    \hline
    & Independent & Random & Cluster \\
    \hline
    Independent & NA & 0.053$^\Delta$ & 0.062$^*$ \\ 
    \hline
    Random & 0.053$^\Delta$ & NA & 0.032 \\ 
    \hline
    Cluster & 0.062$^*$ & 0.032 & NA \\ 
    \hline
    \end{tabular}
    \caption{Pairwise Kolmogorov-D statistic from 2-sample K-S test, $p<0.10 \, (^\Delta)$, $p<0.05 \, (^*)$}
    \label{tab:realdatastoppingtimekstest}
\end{table}

From Tables \ref{tab:realdatastoppingtime} and \ref{tab:realdatastoppingtimekstest}, we observe that the average TTFC for the independent set sample is larger than the other two samples. Moreover, the 2-sample K-S test highlights that the empirical distribution of the TTFC based on the independent set sampling technique is significantly different from that obtained using the random and cluster sampling techniques. This is due to the absence of any edges within the independent set sample which eliminates any direct spillover among the policy respondents for an extended period of time, thus limiting any indirect policy network effect among these respondents. This longer duration of TTFC for the independent set sample, as compared to random or cluster sample, affords policymakers greater flexibility to estimate the direct policy effects. However, we acknowledge that beyond a short term, new edges might develop within the independent set sample due to a variety of reasons e.g., as a result of homophily based on focal behavior or other latent factors. \\

\textbf{Incorporating various centrality measures.} As discussed earlier, one could sample different independent sets for a given social network. Hence, in addition to studying a simple iterative maximal independent set, it is also important to analyze the effects of various independent set selection techniques on the TTFC, as described above. For instance, we can seek to maximize the total or average centrality of the independent set sample based on popular centrality measures. For example, we can maximize the degree sum of all the vertices in the independent set, whose size can be chosen with reference to the bounds of $\alpha(G)$ as discussed in Subsection \ref{sec:indpsetsampling}. Moreover, policymakers may also have certain size or budget constraints. To meet these contextual requirements, we can select an independent set by solving the following integer program, with $A_G$ as the adjacency matrix of the graph $G(V,E)$, $x$ as the vertex incidence vector, $d$ as the degree vector and $c$ as the cost vector of selecting a particular respondent into the sample.
\begin{align}\label{program1}
    \max \; & d^Tx \\ 
    \text{s.t. }& \mathbf{1}^Tx \leq m \nonumber \\ 
    & c^Tx \leq b \nonumber \\
    & x^TA_Gx = 0 \nonumber \\
    & x \in \{0,1\}^{|V(G)|} \nonumber
\end{align}
where $m$ is the size of the independent set, $b$ is the budget, $|V(G)|$ is the number of vertices in the network $G$ and $\mathbf{1}$ is an all-one vector. Suppose policymakers require an independent set of a size $m$ smaller than the lower bound of $\alpha(G)$. In that case, there will be a feasible solution to this integer programming problem and they can find such an independent set $I$ with the largest degree sum, where $I$ contains vertex $i$ if $x_i = 1$. We can change the maximization objective function in the above integer program to incorporate other centrality measures such as closeness, betweenness, page rank, etc., depending on the structural properties that need to be preserved in the policy evaluation exercise. \\

As solving the abovementioned integer program can be computationally intensive, we provide a simpler alternative method to construct the sample under scenarios where the cost of selecting any respondent is the same, and there are no budgetary considerations. As the network interference depends on how central the sampled respondents are, we consider independent sets with high degree and closeness centrality. Instead of solving the integer programs, we employ greedy algorithms (see Algorithms \ref{alg: indpsetdeg} and \ref{alg: indpsetcloseness} in Appendix \ref{appendix:selection} for implementational details) to select independent sets with high degree and closeness centrality, and repeat our simulation exercise. 
The average sample size across the 1000 runs is 35 and 39 for the independent set samples with high degree and closeness centralities, respectively. We illustrate the networks obtained through the various sampling strategies in Figure \ref{fig: graphvisualization} in Appendix \ref{appendix: samplevisualization}. We highlight the differences in average TTFC scores among these independent sets, their respective standard errors, and the pairwise Kolmogorov-Smirnov (K-S) test statistics in Tables \ref{tab:realdatastoppingtimevar} and \ref{tab:realdatastoppingtimekstestvar} below.

\begin{table}[H]
    \centering
    \begin{tabular}{|c|c|c|}
    \hline
    Sampling technique & Empirical average \textit{time to first change} (epoch) & Standard error \\
    \hline
    Independent & 439.477 & 6.781 \\ 
    \hline
    \makecell[c]{Independent \\ (High Degree)} & 456.912 & 6.734 \\ 
    \hline
    \makecell[c]{Independent \\ (High Closeness)} & 454.179 & 6.654 \\ 
    \hline
    \end{tabular}
    \caption{Time to first change (TTFC) across three variants of independent set}
    \label{tab:realdatastoppingtimevar}
\end{table}

\begin{table}[H]
    \centering
    \begin{tabular}{|c|c|c|c|}
    \hline
    & Independent & \makecell[c]{Independent \\ (High Degree)} & \makecell[c]{Independent \\ (High Closeness)} \\
    \hline
    Independent & NA & 0.081$^*$ & 0.078$^\Delta$ \\ 
    \hline
    \makecell[c]{Independent \\ (High Degree)} & 0.081$^*$ & NA & 0.028 \\ 
    \hline
    \makecell[c]{Independent \\ (High Closeness)} & 0.078$^\Delta$ & 0.028 & NA \\ 
    \hline
    \end{tabular}
    \caption{Pairwise Kolmogorov-D statistic from 2-sample K-S test, $p<0.10 \, (^\Delta)$, $p< 0.05 \, (^*)$}
    \label{tab:realdatastoppingtimekstestvar}
\end{table}

From Tables \ref{tab:realdatastoppingtimevar} and \ref{tab:realdatastoppingtimekstestvar}, we observe that the average TTFC of the high-degree and high-closeness independent set samples are slightly higher than the vanilla maximal independent set sample, and that both empirical distributions are highly similar. In Figure \ref{fig: adoption}, we demonstrate the change in the normalized adoption rate of the focal behavior over time, across the 5 sampling strategies, based on 200 runs of our model.

\begin{figure}[H]
    \centering
    \includegraphics[width=0.75\linewidth]{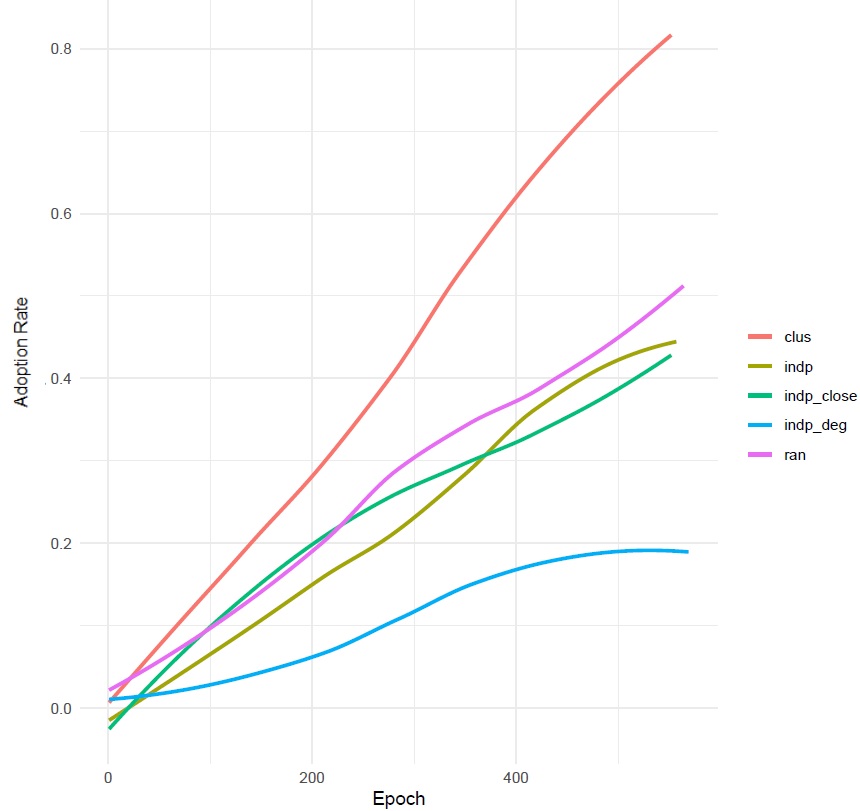}
    \caption{Smoothed adoption rate for different sampling strategies}
    \label{fig: adoption}
\end{figure}

This result complements the observed differences in empirical TTFC across the different sampling strategies. Specifically, we see that the rate of adoption of the focal behavior within the cluster and random samples are higher than the independent set sample. Hence, it is likely that these samples experience network interference faster than the independent set counterparts.\\


\textbf{Estimating direct policy effect}. 
Based on the estimated TTFC, we can attempt to decouple the \textit{direct} effect from the \textit{net} treatment effect of the policy, which includes both direct and network effects. As the above-mentioned sampling strategies lead to different network structures within the sampled treatment groups, we hypothesize that the three sampling conditions will lead to varying estimates of the net effect of the policy. We use the SAOM estimates listed in Table \ref{table:networkestimatesSAOM} to simulate future network-behavior states over 1000 runs, and use the counterfactual approach described in Subsection \ref{sec:estimateeffect} to derive the empirical mean treatment effect of the policy. Since the estimated empirical average TTFC, as shown in Table \ref{tab:realdatastoppingtime}, provides an indication of when network interference might take effect, this can be used as a temporal threshold around which the direct and net treatment effect can be inferred. We elucidate the efficacy of the independent set sampling strategy in the presence of network interference, by tracking the changes in the treatment effect over time under different sampling conditions. As an illustration, we compute the change in treatment effect across two temporal windows based on the estimated average TTFC in Tables \ref{tab:realdatastoppingtime} and \ref{tab:realdatastoppingtimevar}, which demonstrate the gradual onset of network interference after the policy intervention: \textit{first interval} (pre-TTFC vs. short-term around the TTFC; e.g. epochs 200 vs. 425), and \textit{second interval} (short-term around the TTFC vs. long-term post-TTFC; e.g., epochs 425 vs. 490). We repeat the simulations over 1000 runs and present the results of the change in mean treatment effects in Table \ref{tab:gradientfromtreatment} below.

\begin{table}[H]
    \centering
    \begin{tabular}{|c|c|c|}
    \hline
    \multirow{2}{*}{Sampling technique} & \multicolumn{2}{c|}{Changes in treatment effect}\\
    \cline{2-3}
    & First interval & Second interval \\ 
    \hline
    Independent & 0.076 & 0.037\\ 
    \hline
    Random & 0.091 & 0.036\\ 
    \hline
    Cluster & 0.086 & 0.046\\ 
    \hline
    \makecell[c]{Independent \\ (High Degree)} & 0.057 & 0.025 \\ 
    \hline
    \makecell[c]{Independent \\ (High Closeness)} & 0.059 & 0.024 \\ 
    \hline
    \end{tabular}
    \caption{Changes in treatment effect across intervals}
    \label{tab:gradientfromtreatment}
\end{table}

Through the construction of the independent set, we minimize the spillover among the treated respondents and are thus able to decouple the direct and indirect network effects, before the empirically obtained TTFC thresholds, as stated in Table \ref{tab:realdatastoppingtime}. 
From Table \ref{tab:gradientfromtreatment}, we note that the gradient of the mean treatment effect for the independent set strategy is the smallest at 0.076, as compared to 0.091 and 0.086 for the random and cluster sampling strategies respectively, in the first interval. The larger difference in the gradients between the random/cluster and independent set sampling strategies suggests a stronger presence of the network interference effect in the random and cluster samples. The inference is also complemented by the smaller gradient of the independent sets with high degree and closeness, as shown in Table \ref{tab:gradientfromtreatment}. Expectedly, we observe that the gradient across all sampling strategies in the second interval is similar, as a result of the increased network interference over time across all types of samples.

\subsection{Illustrative study using random graph-based models}\label{sec:simulationgraph}
We further extend this stimulation study by analyzing the sensitivity of our proposed method to different graph structures such as small-world and scale-free. We employ three commonly used random graph-based models to generate the underlying population networks, specifically the Erd\"{o}s-Renyi \parencite{erdHos1960evolution}, Watts-Strogatz \parencite{watts1998collective} and scale-free network models \parencite{goh2001universal}. We follow a similar simulation procedure as mentioned in the earlier sections but replace the networks adopted from the \emph{Teenage Friends and Lifestyle Study} \parencites{michell2000smoke, pearson2003drifting} with generated networks from these stated models. We first generate an initial network and include a small perturbation to induce network dynamics, specifically a 0.1\% probability of forming or deleting ties in the network, over the first 3 observations. The choice of 0.1\% probability follows from the noise term in the earlier simulation, and the mean number of edge formation and deletion is similar to what is observed in the \emph{Teenage Friends and Lifestyle Study}.\\

After setting up the simulation specifications, we specify the parameters for the three generative graph models. We note that the graphical structure in the Erd\"{o}s-Renyi (E-R) graph changes once the probability $p(n)$, which is dependent on the number of nodes $n$, exceeds a certain threshold \parencite{jackson2008social}. For example, a giant component starts to emerge at the threshold of $p(n) = \frac{1}{n}$ in the limit. Hence, we run the numerical experiments over different thresholds, specifically, $\frac{1}{n}, \frac{1}{n^{0.5}}$ and $\frac{1}{n^{0.25}}$. For the Watts-Strogatz (W-S) model, we consider re-wiring probabilities of 0.05, 0.10 and 0.15. Whereas, for the scale-free network model (S-F), we consider scale factors of 2, 2.5 and 3 as many scale-free networks have scaling exponents between 2 and 3 \parencites{goh2002classification, chen2004modeling, goh2001universal, barabási2016network}. We set the size of the network for both W-S and S-F models to be 774, which is comparable to the expected size of the E-R graph with edge probability $\frac{1}{n^{0.5}}$. We repeat the simulation over 500 runs to obtain the empirical average TTFC and the respective standard errors, as shown below in Tables \ref{tab:stoppingtimeer}, \ref{tab:stoppingtimews} and \ref{tab:stoppingtimeba}. 

\begin{table}[H]
    \centering
    \begin{tabular}{|c|c|c|c|}
    \hline
    \diagbox[width=15.25em]{Sampling technique}{Probability} & $\frac{1}{n}$ & $\frac{1}{n^{0.5}}$ & $\frac{1}{n^{0.25}}$ \\
    \hline
    Independent & 58.244 (1.806) & 87.580 (1.963) & 116.728 (1.720)\\ 
    \hline
    Random & 62.988 (1.839) & 83.330 (1.968) & 115.344 (1.824) \\ 
    \hline
    Cluster & 64.808 (1.923) & 83.550 (1.999) & 111.484 (1.962) \\ 
    \hline
    \makecell[c]{Independent \\ (High Degree)} & 59.228 (1.835) & 92.606 (1.917) & 121.212 (1.597) \\ 
    \hline
    \makecell[c]{Independent \\ (High Closeness)} & 61.076 (1.860) & 96.138 (1.870) & 121.502 (1.608) \\ 
    \hline
    \end{tabular}
    \caption{Time to first change (TTFC) across three sampling techniques and two variants of independent set sample for E-R graphs}
    \label{tab:stoppingtimeer}
\end{table}

\begin{table}[H]
    \centering
    \begin{tabular}{|c|c|c|c|}
    \hline
    \diagbox[width=15.25em]{Sampling technique}{Re-wiring probability} & 0.05 & 0.10 & 0.15 \\
    \hline
    Independent & 92.132 (1.806) & 88.614 (1.904) & 100.354 (2.255)\\ 
    \hline
    Random & 88.456 (1.794) & 86.714 (1.917) & 97.268 (2.029) \\ 
    \hline
    Cluster & 89.604 (1.907) & 68.560 (1.927) & 89.628 (2.253) \\ 
    \hline
    \makecell[c]{Independent \\ (High Degree)} & 100.286 (1.703) & 95.064 (1.809) & 111.97 (1.743) \\ 
    \hline
    \makecell[c]{Independent \\ (High Closeness)} & 93.214 (1.920) & 93.852 (1.821) & 107.474 (1.883) \\ 
    \hline
    \end{tabular}
    \caption{Time to first change (TTFC) across three sampling techniques and two variants of independent set sample for W-S graphs}
    \label{tab:stoppingtimews}
\end{table}

\begin{table}[H]
    \centering
    \begin{tabular}{|c|c|c|c|}
    \hline
    \diagbox[width=15.25em]{Sampling technique}{Scale Factor} & 2 & 2.5 & 3 \\
    \hline
    Independent & 85.410 (2.132) & 88.528 (2.010) & 90.512 (2.161)\\ 
    \hline
    Random & 79.912 (2.241) & 86.394 (2.058) & 87.264 (2.231) \\ 
    \hline
    Cluster & 83.164 (2.228) & 80.848 (2.044) & 86.968 (2.238) \\ 
    \hline
    \makecell[c]{Independent \\ (High Degree)} & 102.894 (2.070) & 97.802 (1.939) & 91.950 (2.242) \\ 
    \hline
    \makecell[c]{Independent \\ (High Closeness)} & 103.958 (2.084) & 92.482 (2.013) & 89.624 (2.267) \\ 
    \hline
    \end{tabular}
    \caption{Time to first change (TTFC) across three sampling techniques and two variants of independent set sample for S-F graphs}
    \label{tab:stoppingtimeba}
\end{table}

From Tables \ref{tab:stoppingtimeer}, \ref{tab:stoppingtimews} and \ref{tab:stoppingtimeba}, we observe that the estimates of the TTFC are generally consistent with the trends shown in Tables \ref{tab:realdatastoppingtime} and \ref{tab:realdatastoppingtimevar} in the earlier subsection, highlighting that the independent sets are able to delay network interference from taking effect. This suggests that the benefits of an independent set sampling are robust to both changes in population network structure as well as the composition of the independent set. However, there is an exception for the Erd\"{o}s-Renyi model with edge probability $\frac{1}{n}$, where we observe a highly similar empirical average TTFC across all sampling strategies.
We contend that this might be due to the sparse network of about 220 edges. This would lead to minimal edges within the random and cluster samples. The TTFC will naturally be comparable across sampling strategies for such sparse networks. This also suggests that the network size, and particularly the extent of sparsity, might be a factor in determining the effectiveness of the sampling strategies, and future extensions of this work can systematically study this.

\section{Discussion and Conclusion}\label{sec:discussion}
\subsection{Summary of key findings}
Evaluating the impact of social and economic policies is critical to organizations and governments worldwide. The problem of policy evaluation in a networked context has remained under-studied, owing to various methodological limitations with existing work. The question of disentangling the direct impact of a policy is not merely an analytical exercise - it is a linchpin for resource allocation and provides nuanced insights into the interplay between policy effects and network dynamics. In the current study, we propose an empirical strategy that combines a novel sampling strategy based on independent set selection with an SAOM to disentangle the direct policy effect from the net treatment effect, which is a combination of the direct and network-related effects of the policy. We contend that our proposed method offers a context- and policy-agnostic empirical strategy to estimate and disentangle policy effects, and is a good alternative to the commonly used cluster sampling-based techniques, especially in contexts where the selection of a diverse sample of respondents might be preferable. Hence, the overall empirical strategy presented in this paper, as well as the underlying methods and specifications are likely to be useful for decision-makers in a wide variety of organizational and public policy contexts. \\

The results from our simulation exercise demonstrate how different sampling techniques (i.e., independent set vs. random vs. cluster) influence the policy impact on policy-exposed respondents through changes in both the network structure and their behavior. Through the empirical estimation of temporal thresholds (i.e., TTFC), we are able to decouple the direct policy effects from the indirect network effects in the short term. Specifically, our empirical results confirm that it takes a longer average time for the network interference to influence the policy-exposed respondents in the independent set sample, as compared to the random or cluster sets, and that this improvement is robust to both the population network structure as well as the structural composition of the independent set. This longer TTFC can allow policymakers to flexibly measure the direct effects of a policy implementation. To estimate the policy effects, we adopt a counterfactual approach which compares the evolution of the network and behavior with vs. without a policy implementation. The treatment effect can subsequently be inferred at specific points in time (e.g., short-term around the TTFC vs. long-term post-TTFC) as the difference in the proportion of individuals having the focal behavior in the two networks. \\

As mentioned above, we also highlight the robustness of our proposed strategy to variations in the composition of the independent set as well as the contact network structure. 
We also investigate different contact networks generated from random graph-based models using a range of parameters. Our results confirm that the independent set sampling technique leads to consistently longer TTFC across network structures. Furthermore, we note that the network size and sparsity might also play a significant role in determining the effectiveness of the sampling strategies. \\

Understanding the role of sampling strategies in policy design has clear practical implications for policymakers with regard to both implementation cost, as well as their varying effects on the selection of respondents. Based on the design framework and comparative analyses presented in this paper, policymakers can better infer the direct and net effects of a policy. Our specific sample selection technique allows for an optimized policy implementation based on available resources, and helps preserve the diversity of respondents.

\subsection{Limitations and future extensions}
Since this study is among the first to analyze the utility of network sampling strategies in a policy evaluation context, it is important to acknowledge certain limitations, modeling assumptions, and challenges. Firstly, our proposed empirical strategy requires access to the complete contact network to construct an independent set, which might be infeasible in many policy contexts due to resource limitations, legal constraints, and/or regulatory issues. We do note that this is a common limitation across several sampling strategies, although certain techniques such as random sampling may be more immune to this. Furthermore, the construction of the independent set is sensitive to the underlying contact network as the same set of respondents may not form an independent set in a different network. It is therefore imperative to study the robustness and sensitivity of our proposed strategy to data incompleteness and network changes. One possible way to test this is to relax the design of an independent set by allowing a small number of edges to exist within the set. However, we may then require careful adjustments to account for the spread of the policy effects within the treatment group. In addition, our empirical strategy focuses on the structural properties of the graph to isolate the direct policy effects from the indirect network effects. However, this may come at the cost of sample representativeness. For example, strategies such as random or independent set sampling are unable to guarantee a certain proportion of respondents having the focal behavior. Future work can consider exploring stratified versions of our proposed sampling technique to address this challenge. \\

Secondly, in the absence of any prior experimental datasets based on independent set sampling, we resort to a stylized simulation study in this paper to test the relative effectiveness of the sampling strategies. Future empirical studies performed in a real-world experimental or quasi-experimental context can seek to replicate and verify the effectiveness of our proposed sampling technique. We note that there can be right-censoring issues in certain real-world longitudinal studies where respondents may leave the study before they change their focal behaviors. The presence of such data censoring will likely affect the estimation of the TTFC. Thirdly, we make certain simplifying assumptions in this study, such as exposing all the sampled respondents to the policy simultaneously, and assuming that the change in focal behavior occurs instantaneously after the policy implementation and not after a certain lag. Although there are policies or interventions with minimal lags such as regulatory compliance or short-term economic incentives (e.g., limited-time product promotions), we acknowledge that policymakers may choose to implement their policy in batches and there may be a considerable time lag between the policy implementation and the change in focal behavior, as common in public health campaigns and training programs. These lags can confound with policy effects especially when both direct and network effects happen over similar time scales. For such policy contexts, it will be increasingly challenging to isolate the direct policy effect from the associated network effects. Hence, future work can seek to extend our current model by investigating how weakening some of these assumptions can affect the effectiveness of our empirical strategy. Fourthly, though our paper proposes an alternative sampling strategy, we acknowledge that implementing a policy on any group of people, regardless of the seeding strategy, requires careful execution. In our analysis, we assume that the respondents are largely policy-abiding and do not exhibit non-compliant or anomalous behavior. However, future work can investigate the sensitivity of the proposed model to varying levels of non-compliance which is routinely observed in policy contexts. Lastly, all our models rely on a standard Markovian assumption of the data. Although this is a common assumption in SAOM-based studies, it may be a strict assumption for certain study contexts where there are external shocks affecting either the network or the behavior. As a future work, policymakers can incorporate Dynamic Network Actor Models (DyNAMs) to capture and analyze sequence of behavioral changes among the sampled respondents. By analyzing how individual behavior and relationships change over time in response to policy interventions, DyNaMs can provide deeper insights into the mechanisms driving these behavioral shifts.

\printbibliography

@inproceedings{blelloch2012greedy,
  title={Greedy sequential maximal independent set and matching are parallel on average},
  author={Blelloch, Guy E and Fineman, Jeremy T and Shun, Julian},
  booktitle={Proceedings of the twenty-fourth annual ACM symposium on Parallelism in algorithms and architectures},
  pages={308--317},
  year={2012}
}

@article{erdHos1960evolution,
  title={On the evolution of random graphs},
  author={Erd{\"{o}}s, Paul and R{\'e}nyi, Alfr{\'e}d},
  journal={Publ. math. inst. hung. acad. sci},
  volume={5},
  number={1},
  pages={17--60},
  year={1960}
}

@article{watts1998collective,
  title={Collective dynamics of ‘small-world’networks},
  author={Watts, Duncan J and Strogatz, Steven H},
  journal={nature},
  volume={393},
  number={6684},
  pages={440--442},
  year={1998},
  publisher={Nature Publishing Group}
}

@article{chen2004modeling,
  title={The modeling of scale-free networks},
  author={Chen, Qinghua and Shi, Dinghua},
  journal={Physica A: Statistical Mechanics and its Applications},
  volume={335},
  number={1-2},
  pages={240--248},
  year={2004},
  publisher={Elsevier}
}

@article{goh2002classification,
  title={Classification of scale-free networks},
  author={Goh, Kwang-Il and Oh, Eulsik and Jeong, Hawoong and Kahng, Byungnam and Kim, Doochul},
  journal={Proceedings of the National Academy of Sciences},
  volume={99},
  number={20},
  pages={12583--12588},
  year={2002},
  publisher={National Acad Sciences}
}

@article{goh2001universal,
  title={Universal behavior of load distribution in scale-free networks},
  author={Goh, K-I and Kahng, Byungnam and Kim, Doochul},
  journal={Physical review letters},
  volume={87},
  number={27},
  pages={278701},
  year={2001},
  publisher={APS}
}

@book{barabási2016network,
  title={Network Science},
  author={Barab{\'a}si, A.L.},
  isbn={9781107076266},
  year={2016},
  publisher={Cambridge University Press}
}

@article{centola2007complex,
  title={Complex contagions and the weakness of long ties},
  author={Centola, Damon and Macy, Michael},
  journal={American journal of Sociology},
  volume={113},
  number={3},
  pages={702--734},
  year={2007},
  publisher={The University of Chicago Press}
}

@article{hernandez2021environmental,
  title={Environmental stress destabilizes microbial networks},
  author={Hernandez, Damian J and David, Aaron S and Menges, Eric S and Searcy, Christopher A and Afkhami, Michelle E},
  journal={The ISME journal},
  volume={15},
  number={6},
  pages={1722--1734},
  year={2021},
  publisher={Oxford University Press}
}

@article{miller2020experimental,
  title={Experimental and quasi-experimental designs in implementation research},
  author={Miller, Christopher J and Smith, Shawna N and Pugatch, Marianne},
  journal={Psychiatry research},
  volume={283},
  pages={112452},
  year={2020},
  publisher={Elsevier}
}

@article{nadini2021mapping,
  title={Mapping the NFT revolution: market trends, trade networks, and visual features},
  author={Nadini, Matthieu and Alessandretti, Laura and Di Giacinto, Flavio and Martino, Mauro and Aiello, Luca Maria and Baronchelli, Andrea},
  journal={Scientific reports},
  volume={11},
  number={1},
  pages={20902},
  year={2021},
  publisher={Nature Publishing Group UK London}
}

@book{garey1979computers,
  title={Computers and intractability},
  author={Garey, Michael R and Johnson, David S},
  volume={174},
  year={1979},
  publisher={freeman San Francisco}
}

@book{jackson2008social,
  title={Social and economic networks},
  author={Jackson, Matthew O and others},
  volume={3},
  year={2008},
  publisher={Princeton university press Princeton}
}

@article{clauset2004finding,
  title={Finding community structure in very large networks},
  author={Clauset, Aaron and Newman, Mark EJ and Moore, Cristopher},
  journal={Physical review E},
  volume={70},
  number={6},
  pages={066111},
  year={2004},
  publisher={APS}
}

@article{kalish2020stochastic,
  title={Stochastic actor-oriented models for the co-evolution of networks and behavior: An introduction and tutorial},
  author={Kalish, Yuval},
  journal={Organizational Research Methods},
  volume={23},
  number={3},
  pages={511--534},
  year={2020},
  publisher={Sage Publications Sage CA: Los Angeles, CA}
}

@article{michell2000smoke,
  title={Smoke rings: social network analysis of friendship groups, smoking and drug-taking},
  author={Michell, Michael Pearson, Lynn},
  journal={Drugs: education, prevention and policy},
  volume={7},
  number={1},
  pages={21--37},
  year={2000},
  publisher={Taylor \& Francis}
}

@article{pearson2003drifting,
  title={Drifting smoke rings},
  author={Pearson, Michael and West, Patrick},
  journal={Connections},
  volume={25},
  number={2},
  pages={59--76},
  year={2003},
  publisher={Citeseer}
}

@inproceedings{krivitsky2019stergm,
  title={STERGM-Separable Temporal ERGMs for modeling discrete relational dynamics with statnet},
  author={Krivitsky, Pavel N and Goodreau, Steven M},
  year={2019},
  organization={Citeseer}
}

@article{koskinen2015simultaneous,
  title={Simultaneous modeling of initial conditions and time heterogeneity in dynamic networks: An application to foreign direct investments},
  author={Koskinen, Johan and Caimo, Alberto and Lomi, Alessandro},
  journal={Network Science},
  volume={3},
  number={1},
  pages={58--77},
  year={2015},
  publisher={Cambridge University Press}
}

@inproceedings{kempe_2003,
author = {Kempe, David and Kleinberg, Jon and Tardos, \'{E}va},
title = {Maximizing the Spread of Influence through a Social Network},
year = {2003},
publisher = {Association for Computing Machinery},
booktitle = {Proceedings of the Ninth ACM SIGKDD International Conference on Knowledge Discovery and Data Mining},
pages = {137–146},
numpages = {10}
}

@article{snijders2007,
author = {Snijders, Tom and Steglich, Christian and Schweinberger, Michael},
year = {2007},
month = {01},
pages = {41-71},
title = {Modeling the co-evolution of networks and behavior},
journal = {Longitudinal Models in the Behavioral and Related Sciences}
}

@Inbook{Snijders2017,
title="Siena: Statistical Modeling of Longitudinal Network Data",
bookTitle="Encyclopedia of Social Network Analysis and Mining",
year="2017",
publisher="Springer New York",
pages="1--9",
author="Snijders, Tom"
}

@article{snijder_2010,
title = {Introduction to stochastic actor-based models for network dynamics},
journal = {Social Networks},
volume = {32},
number = {1},
pages = {44-60},
year = {2010},
note = {Dynamics of Social Networks},
author = {Tom A.B. Snijders and Gerhard G. {van de Bunt} and Christian E.G. Steglich}
}

@article{blume2015linear,
  title={Linear social interactions models},
  author={Blume, Lawrence E and Brock, William A and Durlauf, Steven N and Jayaraman, Rajshri},
  journal={Journal of Political Economy},
  volume={123},
  number={2},
  pages={444--496},
  year={2015},
  publisher={University of Chicago Press Chicago, IL}
}

@article{kline2012,
title = {Some Interpretation of the Linear-In-Means Model of Social Interactions},
year={2014},
author = {Brendan Kline and Elie Tamer}
}

@article{manski1993,
 author = {Charles F. Manski},
 journal = {The Review of Economic Studies},
 number = {3},
 pages = {531--542},
 title = {Identification of Endogenous Social Effects: The Reflection Problem},
 volume = {60},
 year = {1993}
}

@article{johnsson2017,
title={Estimation of peer effects in endogenous social networks: Control function approach},
  author={Johnsson, Ida and Moon, Hyungsik Roger},
  journal={The Review of Economics and Statistics},
  volume={103},
  number={2},
  pages={328--345},
  year={2021},
  publisher={MIT Press}
}

@article{JOCHMANS2022,
title = {Peer effects and endogenous social interactions},
journal = {Journal of Econometrics},
year = {2022},
author = {Koen Jochmans},
publisher={Elsevier}
}

@article{hu2013,
author = {Hu, Pili and Lau, Wing},
year = {2013},
month = {08},
pages = {},
journal={arXiv preprint arXiv:1308.5865},
title = {A Survey and Taxonomy of Graph Sampling}
}

@article{blondel2008,
author = {Blondel, Vincent and Guillaume, Jean-Loup and Lambiotte, Renaud and Lefebvre, Etienne},
year = {2008},
month = {04},
pages = {},
title = {Fast Unfolding of Communities in Large Networks},
volume = {2008},
journal = {Journal of Statistical Mechanics Theory and Experiment}
}

@article{FORTUNATO2010,
title = {Community detection in graphs},
journal = {Physics Reports},
volume = {486},
number = {3},
pages = {75-174},
year = {2010},
author = {Santo Fortunato}
}

@inproceedings{ugander2013,
  title={Graph cluster randomization: Network exposure to multiple universes},
  author={Ugander, Johan and Karrer, Brian and Backstrom, Lars and Kleinberg, Jon},
  booktitle={Proceedings of the 19th ACM SIGKDD international conference on Knowledge discovery and data mining},
  pages={329--337},
  year={2013}
}

@article{basse2017,
author = {Basse, Guillaume and Airoldi, Edoardo},
year = {2017},
month = {05},
pages = {136-151},
title = {Limitations of Design-based Causal Inference and A/B Testing under Arbitrary and Network Interference},
volume = {48},
journal = {Sociological Methodology}
}

@article{ugander2020,
  title={Randomized graph cluster randomization},
  author={Ugander, Johan and Yin, Hao},
  journal={Journal of Causal Inference},
  volume={11},
  number={1},
  pages={20220014},
  year={2023},
  publisher={De Gruyter}
}

@article{wasserman1996,
author = {Wasserman, Stan and Pattison, Philippa},
year = {1996},
month = {09},
pages = {401-425},
title = {Logit Models and Logistic Regressions for Social Networks: I. An Introduction to Markov Graphs and p*},
volume = {61},
journal = {Psychometrika}
}

@article{block2016,
author = {Block, Per and Stadtfeld, Christoph and Snijders, Tom},
year = {2016},
month = {11},
number = {1},
title = {Forms of Dependence: Comparing SAOMs and ERGMs From Basic Principles},
volume = {48},
pages={202--239},
journal = {Sociological Methods \& Research},
publisher={Sage Publications Sage CA: Los Angeles, CA}
}

@article{shalizi2011,
author = {Shalizi, Cosma and Thomas, Andrew},
year = {2011},
month = {05},
pages = {211-239},
title = {Homophily and Contagion Are Generically Confounded in Observational Social Network Studies},
volume = {40},
journal = {Sociological methods \& research}
}

@article{Hariton2018,
author = {Hariton, Eduardo and Locascio, Joseph},
year = {2018},
month = {06},
title = {Randomised controlled trials - the gold standard for effectiveness research: Study design: randomised controlled trials},
pages={1716},
volume = {125},
number={13},
journal = {BJOG: An International Journal of Obstetrics \& Gynaecology}
}

@article{snijders1996stochastic,
  title={Stochastic actor-oriented models for network change},
  author={Snijders, Tom AB},
  journal={Journal of mathematical sociology},
  volume={21},
  number={1-2},
  pages={149--172},
  year={1996},
  publisher={Taylor \& Francis}
}

@article{schwarz2021,
author = {Schwarz, Tassilo},
year = {2021},
month = {11},
pages = {},
journal={arXiv preprint arXiv:2111.14263},
title = {Randomized Controlled Trials Under Influence: Covariate Factors and Graph-Based Network Interference}
}

@article{karwa2018,
author = {Karwa, Vishesh and Airoldi, Edoardo},
year = {2018},
month = {10},
pages = {},
journal={arXiv preprint arXiv:1810.08259},
title = {A systematic investigation of classical causal inference strategies under mis-specification due to network interference}
}

@techreport{viviano2022policy,
  title={Policy design in experiments with unknown interference},
  author={Viviano, Davide},
  year={2022},
  institution={working paper}
}

@article{crespo2007determinant,
  title={Determinant factors of FDI spillovers--what do we really know?},
  author={Crespo, Nuno and Fontoura, Maria Paula},
  journal={World development},
  volume={35},
  number={3},
  pages={410--425},
  year={2007},
  publisher={Elsevier}
}

@article{viviano2024policy,
  title={Policy targeting under network interference},
  author={Viviano, Davide},
  journal={Review of Economic Studies},
  pages={rdae041},
  year={2024},
  publisher={Oxford University Press UK}
}

@article{mergoni2021,
author = {Mergoni, Anna and De Witte, Kristof},
title = {Policy evaluation and efficiency: a systematic literature review},
journal = {International Transactions in Operational Research},
volume = {29},
number = {3},
pages = {1337-1359},
year = {2022}
}

@book{west1996introduction,
  title={Introduction to Graph Theory},
  author={West, D.B.},
  year={1996},
  publisher={Prentice Hall}
}

@techreport{wei1981lower,
  title={A lower bound on the stability number of a simple graph},
  author={Wei, VK},
  year={1981},
  institution={Bell Laboratories Technical Memorandum New Jersey}
}

@techreport{caro1979new,
  title={New results on the independence number},
  author={Caro, Yair},
  year={1979},
  institution={Tel-Aviv University}
}

@article{Halldórsson1997,
author = {Halldórsson, Magnús and Radhakrishnan, Jaikumar},
year = {1997},
month = {05},
pages = {145-163},
title = {Greed is Good: Approximating Independent Sets in Sparse and Bounded-Degree Graphs.},
volume = {18},
journal = {Algorithmica}
}

@article{steglich2010,
author = {Steglich, Christian and Snijders, Tom and Pearson, Michael},
year = {2010},
month = {08},
pages = {329 - 393},
title = {Dynamic Networks And Behavior: Separating Selection From Influence},
volume = {40},
journal = {Sociological Methodology}
}

@book{maddala1986limited,
  title={Limited-Dependent and Qualitative Variables in Econometrics},
  author={Maddala, G.S.},
  series={Econometric Society Monographs},
  year={1983},
  publisher={Cambridge University Press}
}

@article{robbins1951,
 author = {Herbert Robbins and Sutton Monro},
 journal = {The Annals of Mathematical Statistics},
 number = {3},
 pages = {400--407},
 publisher = {Institute of Mathematical Statistics},
 title = {A Stochastic Approximation Method},
 volume = {22},
 year = {1951}
}

@TechReport{ripley2023,
    title = {Manual for {Siena} version 4.0},
    author={Ripley, Ruth M and Snijders, Tom AB and Boda, Zs{\'o}fia and V{\"o}r{\"o}s, Andr{\'a}s and Preciado, Paulina},
    year = {2023},
    institution = {Oxford: University of Oxford, Department of
      Statistics; Nuffield College},
    note = {R package version 1.3.14.1.
      https://www.cran.r-project.org/web/packages/RSiena/},
  }

@book{gertler2016impact,
  title={Impact evaluation in practice},
  author={Gertler, Paul J and Martinez, Sebastian and Premand, Patrick and Rawlings, Laura B and Vermeersch, Christel MJ},
  year={2016},
  publisher={World Bank Publications}
}

@article{GRAF2020,
title = {A shot in the dark? Policy influence on cluster networks},
journal = {Research Policy},
volume = {49},
number = {3},
pages = {103920},
year = {2020},
issn = {0048-7333},
author = {Holger Graf and Tom Broekel},
publisher={Elsevier}
}

@article{huggins2001inter,
  title={Inter-firm network policies and firm performance: evaluating the impact of initiatives in the United Kingdom},
  author={Huggins, Robert},
  journal={Research Policy},
  volume={30},
  number={3},
  pages={443--458},
  year={2001},
  publisher={Elsevier}
}

@article{athey2017,
Author = {Athey, Susan and Imbens, Guido W.},
Title = {The State of Applied Econometrics: Causality and Policy Evaluation},
Journal = {Journal of Economic Perspectives},
Volume = {31},
Number = {2},
Year = {2017},
Month = {5},
Pages = {3-32}}

@article{aronow2017,
author = {Peter M. Aronow and Cyrus Samii},
title = {{Estimating average causal effects under general interference, with application to a social network experiment}},
volume = {11},
journal = {The Annals of Applied Statistics},
number = {4},
publisher = {Institute of Mathematical Statistics},
pages = {1912 -- 1947},
year = {2017}
}

@inproceedings{maiya2010online,
  title={Online sampling of high centrality individuals in social networks},
  author={Maiya, Arun S and Berger-Wolf, Tanya Y},
  booktitle={Pacific-Asia Conference on Knowledge Discovery and Data Mining},
  pages={91--98},
  year={2010},
  organization={Springer}
}

@article{butts2008relational,
  title={A relational event framework for social action},
  author={Butts, Carter T},
  journal={Sociological methodology},
  volume={38},
  number={1},
  pages={155--200},
  year={2008},
  publisher={Wiley Online Library}
}

@article{perry2013point,
  title={Point process modelling for directed interaction networks},
  author={Perry, Patrick O and Wolfe, Patrick J},
  journal={Journal of the Royal Statistical Society Series B: Statistical Methodology},
  volume={75},
  number={5},
  pages={821--849},
  year={2013},
  publisher={Oxford University Press}
}

@article{stadtfeld2011analyzing,
  title={Analyzing event stream dynamics in two-mode networks: An exploratory analysis of private communication in a question and answer community},
  author={Stadtfeld, Christoph and Geyer-Schulz, Andreas},
  journal={Social Networks},
  volume={33},
  number={4},
  pages={258--272},
  year={2011},
  publisher={Elsevier}
}

@article{stadtfeld2017dynamic,
  title={Dynamic network actor models: Investigating coordination ties through time},
  author={Stadtfeld, Christoph and Hollway, James and Block, Per},
  journal={Sociological Methodology},
  volume={47},
  number={1},
  pages={1--40},
  year={2017},
  publisher={Sage Publications Sage CA: Los Angeles, CA}
}

@article{ivaniushina2019peer,
  title={Peer influence in adolescent drinking behaviour: a protocol for systematic review and meta-analysis of stochastic actor-based modeling studies},
  author={Ivaniushina, Valeria and Titkova, Vera and Alexandrov, Daniel},
  journal={BMJ open},
  volume={9},
  number={7},
  pages={e028709},
  year={2019},
  publisher={British Medical Journal Publishing Group}
}

@article{valente1999accelerating,
  title={Accelerating the diffusion of innovations using opinion leaders},
  author={Valente, Thomas W and Davis, Rebecca L},
  journal={The Annals of the American Academy of Political and Social Science},
  volume={566},
  number={1},
  pages={55--67},
  year={1999},
  publisher={Sage Publications Sage CA: Thousand Oaks, CA}
}

@article{goldstein2018ethical,
  title={Ethical issues in pragmatic randomized controlled trials: a review of the recent literature identifies gaps in ethical argumentation},
  author={Goldstein, Cory E and Weijer, Charles and Brehaut, Jamie C and Fergusson, Dean A and Grimshaw, Jeremy M and Horn, Austin R and Taljaard, Monica},
  journal={BMC medical ethics},
  volume={19},
  pages={1--10},
  year={2018},
  publisher={Springer}
}

@article{baird2018,
    author = {Baird, Sarah and Bohren, J. Aislinn and McIntosh, Craig and Özler, Berk},
    title = "{Optimal Design of Experiments in the Presence of Interference}",
    journal = {The Review of Economics and Statistics},
    volume = {100},
    number = {5},
    pages = {844-860},
    year = {2018},
    month = {12}
}

@article{leung2020treatment,
  title={Treatment and spillover effects under network interference},
  author={Leung, Michael P},
  journal={Review of Economics and Statistics},
  volume={102},
  number={2},
  pages={368--380},
  year={2020},
  publisher={MIT Press One Rogers Street, Cambridge, MA 02142-1209, USA journals-info~…}
}

@article{robins2023multilevel,
  title={Multilevel network interventions: Goals, actions, and outcomes},
  author={Robins, Garry and Lusher, Dean and Broccatelli, Chiara and Bright, David and Gallagher, Colin and Karkavandi, Maedeh Aboutalebi and Matous, Petr and Coutinho, James and Wang, Peng and Koskinen, Johan and others},
  journal={Social Networks},
  volume={72},
  pages={108--120},
  year={2023},
  publisher={Elsevier}
}

@article{frank2023network,
  title={A network intervention for natural resource management in the context of climate change},
  author={Frank, Kenneth A and Chen, Tingqiao and Brown, Ethan and Larsen, Angela and others},
  journal={Social Networks},
  volume={75},
  pages={55--64},
  year={2023},
  publisher={Elsevier}
}

@article{sciabolazza2020connecting,
  title={Connecting the dots: implementing and evaluating a network intervention to foster scientific collaboration and productivity},
  author={Sciabolazza, Valerio Leone and Vacca, Raffaele and McCarty, Christopher},
  journal={Social Networks},
  volume={61},
  pages={181--195},
  year={2020},
  publisher={Elsevier}
}

@article{forastiere2024causal,
  title={Causal inference on networks under continuous treatment interference},
  author={Forastiere, Laura and Del Prete, Davide and Sciabolazza, Valerio Leone},
  journal={Social Networks},
  volume={76},
  pages={88--111},
  year={2024},
  publisher={Elsevier}
}

@article{sanderson2002evaluation,
  title={Evaluation, policy learning and evidence-based policy making},
  author={Sanderson, Ian},
  journal={Public administration},
  volume={80},
  number={1},
  pages={1--22},
  year={2002},
  publisher={Wiley Online Library}
}

@article{shadish2002experimental,
  title={Experimental and quasi-experimental designs for generalized causal inference.},
  author={Shadish, William R and Cook, Thomas D and Campbell, Donald T},
  year={2002},
  publisher={Houghton, Mifflin and Company}
}

@article{borgatti2003network,
  title={The network paradigm in organizational research: A review and typology},
  author={Borgatti, Stephen P and Foster, Pacey C},
  journal={Journal of management},
  volume={29},
  number={6},
  pages={991--1013},
  year={2003},
  publisher={Elsevier}
}

@article{liu2008foreign,
  title={Foreign direct investment and technology spillovers: Theory and evidence},
  author={Liu, Zhiqiang},
  journal={Journal of development economics},
  volume={85},
  number={1-2},
  pages={176--193},
  year={2008},
  publisher={Elsevier}
}

@article{EcklesKarrerUgander+2017,url = {https://doi.org/10.1515/jci-2015-0021},title = {Design and Analysis of Experiments in Networks: Reducing Bias fromInterference},title = {},author = {Dean Eckles and Brian Karrer and Johan Ugander},pages = {20150021},volume = {5},number = {1},journal = {Journal of Causal Inference},doi = {doi:10.1515/jci-2015-0021},year = {2017},lastchecked = {2024-09-30}}
\newpage

\begin{appendices}
\section{Selection of Independent Set}\label{appendix:selection}
In this section, we present implementational details of the method that we employ to obtain our independent set samples. We adopt an iterative approach, as used in other studies \parencite{blelloch2012greedy}, to obtain the independent set for our simulation. We provide the following definitions to complement our algorithm. 

\begin{definition}
A graph $G'(V',E')$ is a \emph{subgraph} of $G(V,E)$ where $V' \subseteq V$ and $E' \subseteq E$. For a subset of nodes $V' \subseteq V$, the subgraph of $G$ induced by $V'$ is the subgraph $G'(V',E')$, where $E' = E \cap \binom{V'}{2}$. A \emph{vertex-induced subgraph} contains a subset of vertices, coupled with edges whose endpoints are both in the subset.
\end{definition}

Also, we add a size limit parameter into the pseudocode to provide the flexibility of controlling for the size of the independent set. In our simulations, we do not constrain to a specific value, so we set $S = |V(G)|$.

\begin{algorithm}[H]
\caption{Iterative Independent Set Selection}\label{alg: indpset}
\textbf{Input: } Graph $G$, Size $S$ \\
\textbf{Output: } Independent Set $I$
\begin{algorithmic}[1]
\While{$G \neq \emptyset$}
    \State Randomly select a vertex $v \in V(G)$
    \State $I \leftarrow I \cup v$
    \State $G \leftarrow \Tilde{G}$, where $\Tilde{G}$ is the vertex-induced subgraph of $V(G) \backslash (v \cup N(v))$
    \Comment{$N(v)$ is the set of vertices which are adjacent to $v$}
    \If{$|I| > S$}
        \State break
    \EndIf
\EndWhile
\State \Return $I$
\end{algorithmic}
\end{algorithm}

\begin{algorithm}[H]
\caption{Greedy High Degree Independent Set Selection}\label{alg: indpsetdeg}
\textbf{Input: } Graph $G$, Size $S$ \\
\textbf{Output: } Independent Set $I$
\begin{algorithmic}[1]
\State \textbf{Initialize:} Sorted vertex array $D$, with all vertices in decreasing order of their degree centralities
\While{$D \neq \emptyset$}
    \State $I \leftarrow I \cup v$, where $v$ is the first element in $D$
    \State $D \leftarrow D \backslash (v \cup N(v))$
    \If{$|I| > S$}
        \State break
    \EndIf
\EndWhile
\State \Return $I$
\end{algorithmic}
\end{algorithm}

\begin{algorithm}[H]
\caption{Greedy High Closeness Independent Set Selection}\label{alg: indpsetcloseness}
\textbf{Input: } Graph $G$, Size $S$ \\
\textbf{Output: } Independent Set $I$
\begin{algorithmic}[1]
\State \textbf{Initialize:} Sorted vertex array $C$, with all vertices in decreasing order of their closeness centralities
\While{$C \neq \emptyset$}
    \State $I \leftarrow I \cup v$, where $v$ is the first element in $C$
    \State $C \leftarrow C \backslash (v \cup N(v))$
    \If{$|I| > S$}
        \State break
    \EndIf
\EndWhile
\State \Return $I$
\end{algorithmic}
\end{algorithm}

\newpage

\section{Structural Preservation of Maximal Independent Sets}\label{appendix:maxindpset}

To test the robustness of picking any arbitrary maximal independent set for our simulation study, we provide several statistics, such as mean degree, mean transitivity, and the pairwise Jenson-Shannon Divergence score based on the degree distribution of the maximal independent sets generated over 500 iterations.

\begin{figure}[H]
    \centering
    \includegraphics[scale=0.75]{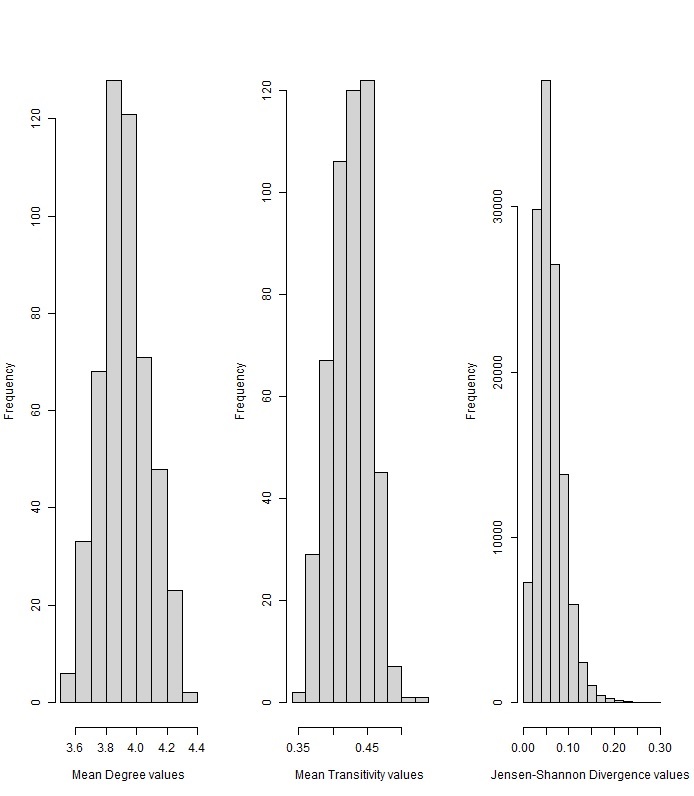}
    \caption{Common network descriptives of the maximal independent sets from dataset}
    \label{fig:glasgowindpsetstats}
\end{figure}

From Figure \ref{fig:glasgowindpsetstats}, we observe that the mean degree of the independent set samples is 3.925, with a standard deviation of 0.156, and a mean transitivity of 0.426, with a standard deviation of 0.028. To check the similarity of the degree distribution between these samples, we compute the Jenson-Shannon Divergence scores. By computing this score for each pairwise degree distribution, we obtain a mean divergence score of 0.058, with a standard deviation of 0.029. \\

\newpage

\section{Graph Visualization of the Samples}\label{appendix: samplevisualization}
We provide a visualization of the group of respondents selected through the respective sampling strategies. The sampled respondents are featured in yellow, whereas the non-sampled others are in grey.

\begin{figure}[H]
     \centering
     \begin{subfigure}[b]{0.3\textwidth}
         \centering
         \includegraphics[width=\textwidth]{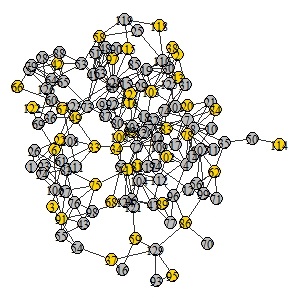}
         \caption{Independent set sample}
         \label{fig: indpsample}
     \end{subfigure}
     \begin{subfigure}[b]{0.3\textwidth}
         \centering
         \includegraphics[width=\textwidth]{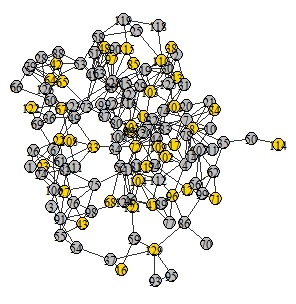}
         \caption{Random sample}
         \label{fig: randomsample}
     \end{subfigure}
     \begin{subfigure}[b]{0.3\textwidth}
         \centering
         \includegraphics[width=\textwidth]{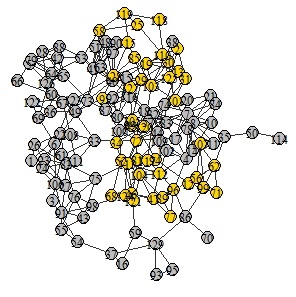}
         \caption{Cluster sample}
         \label{fig: clustersample}
     \end{subfigure}

    \begin{subfigure}[b]{0.3\textwidth}
         \centering
         \includegraphics[width=\textwidth]{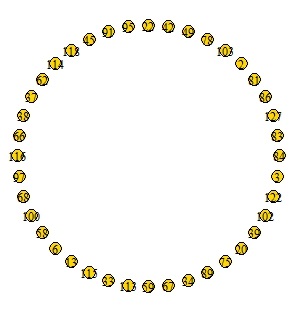}
         \caption{Independent set sample subgraph}
         \label{fig: indpsubgraph}
     \end{subfigure}
    \begin{subfigure}[b]{0.3\textwidth}
         \centering
         \includegraphics[width=\textwidth]{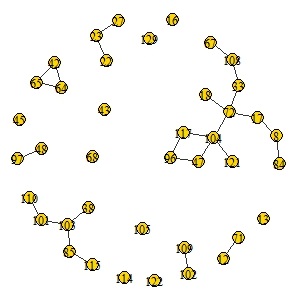}
         \caption{Random sample subgraph}
         \label{fig: randomsubgraph}
     \end{subfigure}
     \begin{subfigure}[b]{0.3\textwidth}
         \centering
         \includegraphics[width=\textwidth]{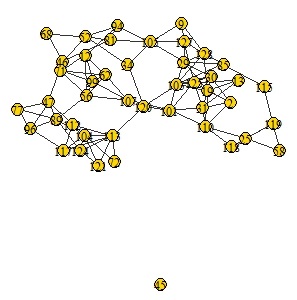}
         \caption{Cluster sample subgraph}
         \label{fig: clustersubgraph}
     \end{subfigure}

    \begin{subfigure}[b]{0.35\textwidth}
         \centering
         \includegraphics[width=\textwidth]{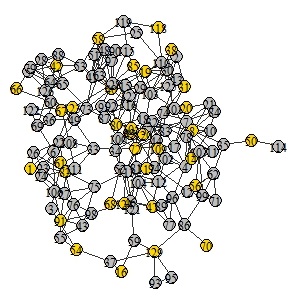}
         \caption{Independent set sample with high degree centrality}
         \label{fig: indpdeg}
     \end{subfigure}
     \quad
     \begin{subfigure}[b]{0.35\textwidth}
         \centering
         \includegraphics[width=\textwidth]{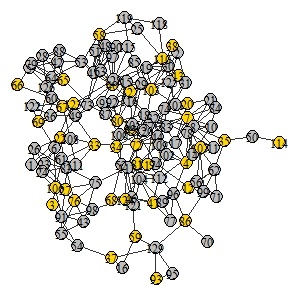}
         \caption{Independent set sample with high closeness centrality}
         \label{fig: indpclose}
     \end{subfigure}
     \caption{Graphs generated using different sampling strategies}
        \label{fig: graphvisualization}
\end{figure}

\newpage

\section{SAOM Formulation}\label{appendix:SAOMdetails}
In this section, we introduce and develop the modeling specifications for the SAOM that we incorporate into our methodology.\\

Adopting from the \emph{Teenage Friends and Lifestyle Study} \parencites{michell2000smoke, pearson2003drifting}, we generate a network of $N=129$ respondents, and model two dependent variables, namely the dynamic friendship network, which is represented by a $N \times N$ symmetric adjacency matrix $A_t$, and a $N \times 2$ behavioral matrix, $B_t$, whose columns correspond to the focal behavior and income, as denoted by $B_{1,t}$ and $B_{2,t}$. We set $B_{1it}$ to be 1 if individual $i$ has a focal behavior at time $t$, and 0 otherwise. \\

We then use an SAOM to specify the co-evolution of network and behavior using a continuous-time Markov process. This enforces the standard Markovian assumption of the conditional distribution of the future being independent of the past, given the current network structure and behavioral characteristics. The model assumes that there is at most one change in the respondents' focal behavior or the number of edges in the network over small time intervals, called \textit{micro-steps} (Steglich et. al., 2010). 
We model the opportunity for any given respondent $i$ at any point in time to form or delete a tie $a_{ij}=[A]_{ij}$, for $j = 1, \dots, i-1, i+1, \dots, N$ or the respondent's focal behavior $b_{1i} = [B_1]_i$ to follow a Poisson process with different rate functions. As we assume that the income of every respondent remains constant after the policy implementation, $b_{2i} = [B_2]_i$ do not undergo any stochastic changes. These functions determine the rates at which respondents make network and behavioral decisions within a time interval. \\ 

For each respondent $i$, there is one rate function for the network, which is denoted by $\lambda^{[A]}_i$, and one for the behavioral changes, which is denoted by $\lambda^{[B]}_i$. Following Steglich et. al. (2010), we describe the rate functions between time periods $t$ and $t+1$ using the following equations
\begin{equation}\label{eqn:networkrate}
\lambda^{[A]}_i(A_t, B_t) = \rho_m^{[A]}\exp{(h_i^{[A]}(\alpha^{[A]},A_t,B_t))}
\end{equation}
\begin{equation}\label{eqn:behaviorrate}
\lambda^{[B]}_i(A_t, B_t) = \rho_m^{[B]}\exp{(h_i^{[B]}(\alpha^{[B]},A_t,B_t))}
\end{equation}
where the parameters $\rho_m^{[A]}$ and $\rho_m^{[B]}$ depend on the time period $m$ and account for periodic changes in either network or focal behavior, and the functions $h_i^{[A]}(\cdot)$ and $h_i^{[B]}(\cdot)$ capture the rate dependence on the current state of network and behavior, with respective weight parameters $\alpha^{[A]}$ and $\alpha^{[B]}$ \parencite{steglich2010}. The exact specifications for $h_i^{[A]}(\cdot)$ and $h_i^{[B]}(\cdot)$ will likely depend on the specific network and behavioral effects that we choose to include for a given network context, and which we specify later for the purpose of our current analysis. \\

The SAOM also relies on objective functions that help to determine which specific changes in network or behavior are to be made at a given micro-step. Each respondent $i$ seeks to optimize this objective function over the set of feasible changes that the respondent can take in the current time period. Snijders et. al. (2007) propose an objective function that consists of three parts, namely the evaluation functions $f_i^{[A]}$ and $f_i^{[B]}$, the endowment functions $g_i^{[A]}$ and $g_i^{[B]}$, and random disturbances $\epsilon_i^{[A]}$ and $\epsilon_i^{[B]}$ \parencite{snijders2007}. \\

The evaluation functions for the network and behavioral decisions, respectively, are parameterized by the vectors $\beta^{[A]}$ and $\beta^{[B]}$; the endowment functions are parameterized by the vectors $\gamma^{[A]}$ and $\gamma^{[B]}$, as shown in the following equations.
\begin{equation}\label{eqn:networkobj}
\text{Network decisions:} \quad f_i^{[A]}(\beta^{[A]}, A_t, B_t) + g_i^{[A]}(\gamma^{[A]}, A_t, B_t|A_{t-1},B_{t-1}) + \epsilon_i^{[A]}(A_t, B_t)
\end{equation}
\begin{equation}\label{eqn:behaviorobj}
\text{Behavioral decisions:} \quad f_i^{[B]}(\beta^{[B]}, A_t, B_t) + g_i^{[B]}(\gamma^{[B]}, A_t, B_t|A_{t-1},B_{t-1}) + \epsilon_i^{[B]}(A_t, B_t)
\end{equation}

The evaluation functions, $f_i^{[A]}$ and $f_i^{[B]}$ measure the respondents' utility based on the current state of the network and their focal behavior, with respective weight parameters $\beta^{[A]}$ and $\beta^{[B]}$. The respondents continuously strive to alter their friendship network and focal behavior to maximize their utility based on the evaluation function. \\

The endowment functions $g_i^{[A]}$ and $g_i^{[B]}$, with their respective weight parameters $\gamma^{[A]}$ and $\gamma^{[B]}$, capture the loss in utility due to a unit change in the network ties or focal behavior, which were gained earlier. 
In other words, these functions can be used to simulate scenarios where the formation and breaking of links, or the changes in focal behavior, generate asymmetric gains or losses for the respondents. In our simulation, we assume that the loss in utility is the same as the respondent's gain from a change. Hence, no endowment functions are specified in our current model.\\

The random noises $\epsilon_i^{[A]}$ and $\epsilon_i^{[B]}$ represent a portion of the respondent's preference, which is not captured by either the evaluation or endowment functions. By assuming that these random noises follow the type-1 extreme value distribution, similar to random utility models, we obtain a closed-form multinomial logit expression for the probabilities of the network and focal behavioral micro-step decisions \parencite{maddala1986limited}. Based on Snijders et. al. (2007), the choice probability that is derived from the network micro-step decisions is given as, 
\begin{equation}\label{eqn:networkprobability}
\Prob{a_{ij,t+1} = a_{ij,t} + \delta|a_t,b_t, \beta^{[A]}} = \dfrac{\exp{(f_i^{[A]}(\beta^{[A]}, a_{ij,t+1} = a_{ij,t} + \delta, b_t))}}{\displaystyle \sum_{k \in [n]\backslash i}\sum_\psi\exp{(f_i^{[A]}(\beta^{[A]}, a_{ik,t+1} = a_{ik,t} + \psi, b_t))}},
\end{equation}
where $a_{t+1}$ is the resulting network at $t+1$ when respondent $i$ at micro-step $t$ either creates a new tie, deletes an existing tie, or makes no change to the respondent's connections \parencite{snijders2007}. We model the change in ties by altering the variables $\delta$ (or $\psi$), where $\delta, \psi \in \{0,\pm 1\}$. Similarly, based on Snijders et. al. (2007), the choice probability that is derived from the focal behavioral micro-step decisions is given as, 
\begin{equation}\label{eqn:behaviorprobability}
\Prob{b_{1i,t+1} = b_{1it} + \delta|a_t,b_t, \beta^{[B]}} = \dfrac{\exp{(f_i^{[B]}(\beta^{[B]}, a_t, b_{1i,t+1} = b_{1i,t} + \delta))}}{\displaystyle \sum_\psi\exp{(f_i^{[B]}(\beta^{[B]}, a_t, b_{1i,t+1} = b_{1it} + \psi))}},
\end{equation}
where $b_{1,t+1}$ is the resulting state of focal behavior at $t+1$ when respondent $i$ at micro-step $t$ either attains, loses, or makes no changes to the respondent's focal behavior \parencite{snijders2007}. We model the change in focal behavior by altering the variables $\delta$ (or $\psi$), where $\delta, \psi \in \{0,\pm 1\}$.\\

After formulating the choice probabilities, the transition intensity matrix $Q$ can be obtained as follows,
\begin{equation}\label{eqn:intensitymatrix}
\text{$Q(a_{t+1}, b_{t+1})$} = 
\begin{cases}
\lambda^{[A]}_i\Prob{a_{ij,t+1} = a_{ij,t} + \delta|a_t,b_t} & \text{if ($a_{ij,t+1}, b_{t+1}$) = ($a_{ij,t} + \delta, b_t$)}\\
\lambda^{[B]}_i\Prob{b_{1i,t+1} = b_{1it} + \delta|a_t,b_t} & \text{if ($a_{t+1}, b_{1i,t+1}$) = ($a_t, b_{1i,t} + \delta$)} \\
\begin{aligned}
-\displaystyle \sum_{i}& \Bigl\{\sum_{j\neq i}\sum_{\delta \in \{-1, 1\}}Q(a_{ij,t+1} +\delta, b_{t+1})\\
& +\sum_{\delta \in \{-1, 1\}}Q(a_{t+1}, b_{1i,t+1}+\delta)\Bigl\}
\end{aligned}& \text{if ($a_{t+1}, b_{t+1}$) = ($a_t, b_t$)}\\
0 & \text{otherwise}
\end{cases}
\end{equation}
This models the rate of transitioning from the state ($a_t, b_t$) at micro-step $t$ to a new state ($a_{t+1}, b_{t+1}$) at micro-step $t+1$. As it is difficult to obtain a closed-form likelihood function, a simulation-based estimator, specifically a Monte Carlo Markov Chain (MCMC)-based Method of Moments (MoM) estimator, is used to obtain the parameters $\beta$, $\rho$, $\alpha$, of the rate and evaluation functions. 
The MCMC implementation of the MoM estimator uses a stochastic approximation algorithm, which is adapted from the Robbins–Monro algorithm \parencite{robbins1951}, and 
implemented in the R package \code{RSiena} \parencite{ripley2023}.\\ 

We model our rate and objective functions $h_i^{[A]}$, $h_i^{[B]}$, $f_i^{[A]}$ and $f_i^{[B]}$ from (\ref{eqn:networkrate}), (\ref{eqn:behaviorrate}), (\ref{eqn:networkobj}) and (\ref{eqn:behaviorobj}) as a weighted sum of relevant network characteristics (e.g., degree, transitivity, and
homophily based on the respondent's covariates), and behavioral characteristics (e.g., linear tendency, peer influence, and the effect of the respondent's connections on focal behavior). We denote the matrix of the network and behavior statistics computed in each time period $t$ by $S_t^{[A]}$ and $S_t^{[B]}$, which are $N \times K$ and $N \times L$ matrices of $K$ network and $L$ behavioral characteristics, respectively. We specify the functions
$h_i^{[A]}$ and $h_i^{[B]}$ from the rate functions as follows,
\begin{equation}\label{eqn:hiA}
h_i^{[A]}(\alpha^{[A]},A_t,B_t) = \sum_{q}\alpha^{[A]}_qs_{iqt}^{[A]}(A_t,B_t)
\end{equation}
\begin{equation}\label{eqn:hiB}
h_i^{[B]}(\alpha^{[B]},A_t,B_t) = \sum_{r}\alpha^{[B]}_rs_{irt}^{[B]}(A_t,B_t),
\end{equation}
where $\alpha^{[A]}_q$ and $\alpha^{[B]}_r$ measure the weights of the respective rate statistics $s_{iqt}^{[A]} = [S_t^{[A]}]_{iq}$ and $s_{irt}^{[B]} = [S_t^{[B]}]_{ir}$, which are one-dimensional vectors defined for each respondent $i$ that capture the rate dependence on the respondent's network and behavioral characteristics, and $q \subset K, r \subset L$. For simplicity, we do not include any rate effects in our study, i.e., we only rely on the basic period-dependent parameter $\rho_m^{[A]}$ and $\rho_m^{[B]}$ for period $m = 1, 2$ with constant functions $h_i^{[A]}$ and $h_i^{[B]}$ for every respondent $i$.\\

Similarly, we specify the functions $f_i^{[A]}$ and $f_i^{[B]}$ from the objective functions as follows,
\begin{equation}\label{fiA}
f_i^{[A]}(\beta^{[A]},A_t,B_t) = \displaystyle \sum_{c}\beta^{[A]}_cs_{ict}^{[A]}(A_t,B_t)
\end{equation}
\begin{equation}\label{eqn:fiB}
f_i^{[B]}(\beta^{[B]},A_t,B_t) = \displaystyle \sum_{d}\beta^{[B]}_ds_{idt}^{[B]}(A_t,B_t),
\end{equation}
where $s_{ict}^{[A]} = [S_t^{[A]}]_{ic}$ and $s_{idt}^{[B]} = [S_t^{[B]}]_{id}$  are the respective $c$-th network statistic and $d$-th behavioral statistic of respondent $i$ and $c \subset K$, $d \subset L$. We specify the following statistics for the network and behavioral effects which we use in our model. The network effects we choose to include are transitivity ($s_{i1t}^{[A]}$), the respondent $i$'s homophily effects based on the focal behavior ($s_{i2t}^{[A]}$) and income ($s_{i3t}^{[A]}$), popularity effect ($s_{i4t}^{[A]}$) and covariate-ego activity based on the focal behavior ($s_{i5t}^{[A]}$) and income ($s_{i6t}^{[A]}$), and are specified as follows \parencite{ripley2023},
\begin{enumerate}
    \item Transitivity ($s_{i1t}^{[A]}$). This measures the triads formed within the network 
    \begin{equation}\label{eqn:transitivity}
        s_{i1t}^{[A]}(a) = \sum_{j,k}a_{ijt}a_{jkt}a_{ikt}
    \end{equation}
    \item Homophily based on focal behavior ($s_{i2t}^{[A]}$) and the income ($s_{i3t}^{[A]}$)
    \begin{equation}\label{eqn:homophily}
    s_{iwt}^{[A]}(a,b) = \sum_ja_{ijt}(sim^{b_{w-1}}_{ijt}-\widehat{sim^{b_{w-1}}_t}), \quad \text{for $w \in \{2,3\}$}
    \end{equation}
    where $\widehat{sim^{b_{w-1}}_t}$ denotes the mean of all similarity scores and $sim^{b_{w-1}}_{ijt} = \frac{\Delta_t - |b_{(w-1)it}-b_{(w-1)jt}|}{\Delta_t}$ with $\Delta_t = \max_{ij}|b_{(w-1)it}-b_{(w-1)jt}|$ as the maximum observed range of the covariate $b_{w-1}$. We note that these variables take on a higher value for those respondents whose focal behavior or income is closer to that of their peers (i.e. the value of $|b_{(w-1)it}-b_{(w-1)jt}|$ is small). A higher value for $s_{iwt}^{[A]}$ is indicative of a heightened propensity towards creating ties based on similar focal behavior.
    \item Popularity effect ($s_{i4t}^{[A]}$), which is defined by the sum of the in-degrees of node $i$'s neighbors.
    \begin{equation}\label{eqn:indeg}
    s_{i4t}^{[A]}(a) = \sum_ja_{ij}\sum_ha_{hj}
    \end{equation}
    \item Covariated-related activity based on focal behavior ($s_{i5t}^{[A]}$) and the income ($s_{i6t}^{[A]}$)
    \begin{equation}\label{eqn:activity}
    s_{ilt}^{[A]}(a) = b_{mit}a_{i^+t}, \quad \text{for $(l, m) \in \{(5,1), (6,2)\}$}
    \end{equation}
\end{enumerate}
In addition to these effects, the density effect is also included in the model by default in \code{RSiena}. \\

Similar to the network effects stated above, we also model a number of important behavioral effects. These effects include the respondent $i$'s behavior tendency effect ($s_{i1t}^{[B]}$), outdegree connection ($s_{i2t}^{[B]}$), the peer influence effect i.e. social influence ($s_{i3t}^{[B]}$) and the income covariate effect ($s_{i4t}^{[B]}$). The specifications for these effects, as listed in \parencite{ripley2023}, are presented as follows:
\begin{enumerate}
    \item Behavioral tendency effect ($s_{i1t}^{[B]}$). This captures the tendency of respondents to attain or lose the focal behavior over time
    \begin{equation}\label{eqn:linearshape}
        s_{i1t}^{[B]}(a,b) = b_{1it}
    \end{equation}
    \item Outdegree effect ($s_{i2t}^{[B]}$). This measures the effect of the respondents' connections on their focal behavior
    \begin{equation}\label{eqn:outdegeffect}
        s_{i2t}^{[B]}(a,b) = b_{1it}\sum_ja_{ijt}
    \end{equation}
    \item Average peer influence effect ($s_{i3t}^{[B]}$). This measures the propensity of respondents to assimilate their focal behaviors toward their peers.
    \begin{equation}\label{eqn:avpeereffect}
        s_{i3t}^{[B]}(a,b) = a_{i^+t}^{-1}\sum_ja_{ijt}(sim^{b_1}_{ijt}-\widehat{sim^{b_1}_t}); \quad (\text{and $0$ if $a_{i^+t}=0$}),
    \end{equation}
    where the construction of $sim^{b_1}_{ijt}$ and $\widehat{sim^{b_1}_t}$ are similar to equation \ref{eqn:homophily}. We note that $s_{i3t}^{[B]}$ has a higher value for those respondents whose focal behavior is closer to that of their peers (i.e. the value of $|b_{1it}-b_{1jt}|$ is small). Hence, a positive and significant estimate of this effect indicates that respondents alter their focal behavior to match their peers and vice versa.
    \item Income covariate effect ($s_{i4t}^{[B]}$).
    \begin{equation}\label{eqn:incomecovariateeffect}
        s_{i4t}^{[B]}(a,b) = b_{2it}b_{1it}
    \end{equation}
\end{enumerate}

\newpage

\end{appendices}

\end{document}